%% file: sample-sigconf.tex
\documentclass{article}

\input{ZZ_header.tex}

\begin{document}

\title{Verifiable Differential Privacy} 


\author[1]{Ari Biswas}
\author[2]{Graham Cormode}
\affil[1]{University Of Warwick}
\affil[2]{Meta AI}



\maketitle

\begin{abstract}
Differential Privacy (DP) is often presented as a strong privacy-enhancing technology with broad applicability and advocated as a de-facto standard for releasing aggregate statistics on sensitive data. 
However, in many embodiments, DP introduces a new attack surface: a malicious entity entrusted with releasing statistics could manipulate the results and use the randomness of DP as a convenient smokescreen to mask its nefariousness. 
Since revealing the random noise would obviate the purpose of introducing it, the miscreant may have a perfect alibi.  
To close this loophole, we introduce the idea of \textit{Verifiable Differential Privacy}, which requires the publishing entity to output a zero-knowledge proof that convinces an efficient verifier that the output is both DP and reliable.
Such a definition might seem unachievable, as a verifier must validate that DP randomness was generated faithfully without learning anything about the randomness itself. 
We resolve this paradox by carefully mixing private and public randomness to compute verifiable DP counting queries with theoretical guarantees and show that it is also practical for real-world deployment. 
We also demonstrate that computational assumptions are necessary by showing a separation between information-theoretic DP and computational DP under our definition of verifiability. 
\end{abstract}

\input{tex_files/introduction.tex}

\input{tex_files/background.tex}

\input{tex_files/new_definitions.tex}
\input{tex_files/single_curator_bin_mechanism.tex}

\input{tex_files/separation.tex}
\input{tex_files/performance}

\input{tex_files/related_work.tex}
\bibliographystyle{alpha}
\bibliography{sample-base}

\appendix

\input{Appendix/a_security_proof.tex}
\input{Appendix/b_dp_proofs}
\input{Appendix/c_fiat_shamir}
\input{Appendix/d_mpc_defs.tex}

\end{document}

%% file: ZZ_header.tex
\usepackage{url}
\usepackage{algorithm}
\usepackage{tikz}
\usepackage{graphicx}
\usepackage{xspace}
\usepackage{subcaption}
\usepackage{authblk}
\usepackage{stmaryrd}
\usepackage{booktabs}
\usepackage[square,sort,comma,numbers]{natbib}
\usepackage{arxiv}

\pdfoutput=1
\usepackage[n, 
adversary , landau , probability , notions ,
logic ,ff, mm,
primitives , events , complexity , oracles , asymptotics , keys]{cryptocode}

\definecolor{darkcyan}{rgb}{0.0, 0.55, 0.55}
\definecolor{brinkpink}{rgb}{0.98, 0.38, 0.5}
\definecolor{lightgreen}{HTML}{90EE90}
\definecolor{cyan(process)}{rgb}{0.0, 0.72, 0.92}

\usepackage[colorlinks=true, linkcolor=cyan(process), ,citecolor={darkcyan}]{hyperref}
\usepackage{xcolor}
\colorlet{linkequation}{cyan}

\newcommand*{\SavedEqref}{}
\let\SavedEqref\eqref
\renewcommand*{\eqref}[1]{%
  \begingroup
    \hypersetup{
      linkcolor=linkequation,
      linkbordercolor=linkequation,
    }%
    \SavedEqref{#1}%
  \endgroup
}

\newcommand{\bin}{\texttt{Bin}}
\newcommand{\morra}{\texttt{morra}}
\newcommand{\etal}{\textit{et al.}{ }}
\newcommand{\Pv}{\mathtt{Pv}}
\newcommand{\Vfr}{\mathtt{Vfr}}
\newcommand{\bit}{\ensuremath{\{0,1\}}}

\newcommand{\share}[1]{\llbracket #1 \rrbracket}

\newcommand{\Sim}{\ensuremath{\mathsf{Sim}}{ }}
\newcommand{\xmark}{\text{}}

\newcommand{\AdvA}{\ensuremath{\mathcal{A}}}
\renewcommand{\oracle}{\ensuremath{\mathcal{O}}\xspace}

\newcommand{\Z}{\mathbb{Z}}

\newcommand{\N}{\mathbb{N}}

\newcommand{\G}{\mathbb{G}}
\renewcommand{\pp}{\texttt{pp}}

\DeclareMathOperator{\Com}{Com}
\newcommand{\noisen}{n_b}
\newcommand{\parag}[1]{\smallskip\noindent\textbf{#1.}}
\DeclareMathOperator{\Bin}{Bin}

\newcommand{\bvec}[1]{\vec{#1}}
\newtheorem{theorem}{Theorem}[section]

\newtheorem{lemma}[theorem]{Lemma}
\newtheorem{definition}{Definition}
\newtheorem{proof}{Proof}

%% file: tex_files/introduction.tex
\pdfoutput=1

\section{Introduction}
\label{sec:introduction}

We are living in an age of delegation, where the bulk of our digital data is held and processed by others in an opaque fashion. 
Our interactions are collated by digital applications that continually send our personal information to the ``cloud''. 
Servers in the cloud, typically owned by large monolithic organizations, such as Google, AWS or Microsoft, then perform computations on our private data to publish aggregate statistics for social utility.
For example, we send our GPS coordinates to services like Strava and Google which, in exchange, use this information to recommend low-traffic cycling routes~\cite{raturi2021impact}. 
Similarly, we let entertainment companies like Netflix, YouTube, TikTok and Hulu know our personal preferences so that they can better recommend content for us to consume~\cite{bell2007lessons}. 
National census bureaus collect personal information to publish aggregate statistics about the population, and consider doing so a moral duty to ensure transparency in the government's policies~\cite{boyd2022differential}. 

However, it is often the case that published aggregate statistics leak information about the activity of individuals. 
For example, Garfinkel \etal and Kasiviswanathan \etal describe practical reconstruction attacks that can be used to infer an individual's private data from aggregate population statistics ~\cite{garfinkel2019understanding, kasiviswanathan2013power}. 
Boyd \etal show that published census data has been used to discriminate against groups in society based on race ~\cite{boyd2022differential}. 
Hence the information that is released, and how it is computed, requires careful scrutiny. 

In response to these concerns, the privacy and security community have sought to apply various privacy enhancing technologies to protect the privacy of individuals contributing to data releases. 
Most relevant to this discussion is Differential Privacy (DP) and its generalizations, which require computations to be randomized, in order to offer the (informally stated) promise that users will not be adversely affected by allowing their data to be used.  
Typically, this is achieved by adding carefully calibrated random noise to the output, at the expense of reducing the accuracy of the computation. 
Differential privacy is most commonly studied in the \textit{trusted curator} model, where a single entity receives all the sensitive data, and is entrusted to execute the algorithm to apply the random noise.  
Variations that modify the trust and computational model include local privacy \cite{warner1965randomized}, shuffle privacy \cite{balle2019privacy, champion2019securely}, computational differential privacy ~\cite{mironov2009computational} and multi-party differential privacy \cite{mcgregor2010limits}. 


A consistent theme across all existing work is to view DP simply as a privacy preserving mechanism. 
In this paper we shift the focus and view differential privacy through an adversarial lens: \textit{what if an adversary seeks to abuse the protocol and pick noise chosen to distort the statistics, using differential privacy as an attack vector?}
That is, a malicious entity may tamper with the computation in order to publish biased statistics, and claim this reflects the true outcome; any discrepancies may be dismissed as artifacts of random noise. 
Consider a counting query DP protocol to determine the winner of a plurality election, where the users vote for 1 out of $M$ candidates (say, which topping people prefer on their pizza). 
A corrupted aggregator might not be interested in any particular user's vote but in biasing the aggregate output of the protocol instead. 
Thus, if that server has auxiliary information about the preferences of a subset of users, they might tamper with the protocol to exclude those honest voters from the election or tamper with the protocol to bias the results of the election (say, to pineapple) and blame any discrepancies in the result on random noise introduced by DP.
Note that some loss in accuracy for privacy is unavoidable.
By definition, DP requires the output be perturbed by private randomness. 
Often, outputting such random statistics creates tensions between publishing entity and the downstream consumer. 
In 2021, the State Of Alabama filed a lawsuit claiming that the use of DP on census data was illegal~\cite{courtCase}, citing the inaccuracies introduced by DP. 
Thus, to ensure public trust in DP, it is critical to verify that any loss in utility can be attributed solely to unavoidable DP randomness.

To that effect, we formally introduce the idea of \textit{publicly verifiable differential privacy} in both the trusted curator setting and the multi-party setting in presence of active adversaries\footnote{By active adversaries, we mean participants that may deviate from protocol specifications arbitrarily}. Our contributions are as follows:

\begin{enumerate}

    \item{We formally introduce \textit{publicly verifiable differential privacy} in both the trusted curator and client-server multiparty setting \cite{baum2014publicly}. Informally, the entity responsible for releasing DP statistics is required to also output a zero-knowledge proof to verify that the statistic was computed correctly and the private randomness generated faithfully. Such a proof reveals no additional information and still enforces that user privacy is protected via DP but ensures that the curator cannot use DP randomness maliciously.}

    \item{We show concrete instantiations of verifiable DP by computing DP counting queries (histograms) in the trusted curator setting and the client-server multiparty setting. In the trusted curator setting, there is a single aggregating server that sees client data in plaintext and is responsible for outputting a DP histogram and a proof that the DP noise was generated faithfully. In the client-server MPC setting, clients secret share the inputs and send them to $K \geq 2$ servers, who then participate in an MPC protocol to output DP histograms. The protocol itself is secure in that not even the participating servers are able to learn any new information beyond the output nor are they able to tamper with the protocol.}
    
	\item{We conduct experiments to show that our protocols with formal theoretical guarantees are also practical. 
   Additionally, we describe how our protocol $\Pi_\bin$, for verifiable DP counting, can be combined with existing (non-verifiable) DP-MPC protocols, such as PRIO \cite{corrigan-gibbs_prio_2017} and Poplar \cite{boneh_lightweight_2022}, to enforce verifiability.}

	\item{We demonstrate that information-theoretic verifiable DP is impossible. Specifically, if both the prover and verifier are computationally unbounded, then both statistical zero knowledge and unconditional soundness cannot hold. Thus we could either prevent an all-powerful curator from manipulating DP protocols or an all-powerful verifier from being able to distinguish between neighbouring datasets from the output, but not both.  This result is related to an open problem (Open Problem 10.6) of Vadhan~\cite{vadhan2017complexity}, which asks \textit{``Is there a computational task solvable by a single curator with computational differential privacy but is impossible to achieve with information-theoretic differential privacy?''}. 
    In Section \ref{sec:separation} we relate our result to efforts at resolving this question.}

\end{enumerate}

%% file: tex_files/background.tex
\pdfoutput=1
\section{Preliminaries}
\label{sec:prelims}

\subsection{Notation}

We write $x \xleftarrow{R} U$ to denote that $x$ was uniformly sampled from a set $U$. We denote vectors with an arrow on top as in $\bvec{x} \in \Z_q^M$, where $M$ represents the number of coordinates in the vector and $\Z_q$ represents a prime order finite field of integers of size $q$. We write $\bvec{a} + \bvec{b}$ to mean coordinate-wise vector addition $a + b \mod q$, where $a$ and $b$ are arbitrary coordinates of $\bvec{a}$ and $\bvec{b}$. 
Similarly, when we write $\bvec{a} \times \bvec{b}$, we refer to the coordinate-wise Hadamard product between the two vectors.

\subsection{Privacy and Security Background}
\label{sec:background}

\parag{Indistinguishability} 
We define a computational notion of indistinguishability. 
\begin{definition}[Computational Indistinguishability]
\label{defn:indisitinguish}   Fix security parameter $\kappa \in \N$. Let $\{ X_\kappa\}_{\kappa \in N }$ and $\{ Y_\kappa\}_{\kappa \in N}$ be probability distributions over $\bit^{\texttt{poly}(\kappa)}$. We say that $\{ X_\kappa\}_{\kappa \in N }$ and $\{ Y_\kappa\}_{\kappa \in N }$ are computationally indistinguishable $\{ X_\kappa\}_{\kappa \in N } \stackrel{c}{\equiv} \{ Y_\kappa\}_{\kappa \in N }$ if for all non-uniform PPT turing machines $D$ (``distinguishers''), there exists a negligible function $\mu(\kappa)$ such for every $\kappa \in \N$

\begin{equation}
    \Big| \Pr[D(X_\kappa) = 1] - \Pr[D(Y_\kappa) = 1] \Big| \leq \mu(\kappa) 
\end{equation}
\end{definition}

\parag{Commitments}\label{sec: commitments} Commitments are used in our schemes to ensure that participants cannot change their response during the protocol.
\begin{definition}[Commitments] Let $\kappa \in \N$ be the security parameter. A non-interactive commitment scheme consists of a pair of probabilistic polynomial time algorithms (\texttt{Setup, Com}). The setup algorithm $\pp \leftarrow \texttt{Setup}(1^\kappa)$ generates public parameters $\pp$. Given a message space $\texttt{M}_{\pp}$ and randomness space $\texttt{R}_{\pp}$, the commitment algorithm $\texttt{Com}_{\pp}$ defines a function $\texttt{M}_{\pp} \times \texttt{R}_{\pp}  \rightarrow \texttt{C}_{\pp}$ that maps a message to the commitment space $\texttt{C}_{\pp}$ using the random space. For a message $x \in \texttt{M}_\pp$, the algorithm samples $r_x \xleftarrow{R} \texttt{R}_\pp$ and computes $c_x = \texttt{Com}_{\pp}(x, r_x)$. When the context is clear, will drop the subscript and write $\texttt{Com}_\pp $ as $\texttt{Com}$.
\end{definition}

\begin{definition}[Homomorphic Commitments]
\label{defn:hom_coms}
A homomorphic commitment scheme is a non-interactive commitment scheme such that $\texttt{M}_\pp$ and $ \texttt{R}_\pp$ are fields (with $(+, \times)$) 
and $\texttt{C}_\pp$ are is an abelian groups with the $\otimes$ operator on which the discrete log problem (Definition \ref{defn:discrete_log}) is hard, so that for all $x_1, x_2 \in \texttt{M}_\pp$ and $r_1, r_2 \in \texttt{R}_\pp$ we have
\begin{equation}
\label{eq:hom_coms}
\texttt{Com}(x_1, r_1) \otimes \texttt{Com}(x_2, r_2) = \texttt{Com}(x_1 + x_2, r_1 + r_2)
\end{equation}
\end{definition}


Throughout this paper, when we use a commitment scheme, we mean a non-interactive homomorphic commitment scheme  with the following properties (stated informally here, but formalized in the Appendix \ref{app:sec_defns}): 

\begin{enumerate}
    \item{\textbf{Hiding:}  A commitment $c_x$ reveals no information about $x$ and $r_x$ to a computationally bounded adversary (Definition~\ref{defn:hiding_com}).}

    \item{\textbf{Binding:} Given a commitment $c_x$ to $x$ using $r_x$, there is no efficient algorithm that can find $x^\prime$ and $r_{x^\prime}$ such that $\Com(x^\prime, r_{x^\prime}) = c_x = \Com(x, r_{x})$ (Definition~\ref{defn:binding_com}).}

    \item{\textbf{Zero Knowledge OR Opening:} Given $c_x$, the commiting party is able to prove to a polynomial time verifier that $c_x$ is a commitment to either 1 or 0 without revealing exactly which one it is. We denote such a proof as $\Pi_{\texttt{OR}}$ and say it securely computes the oracle $\oracle_{\texttt{OR}}$, which computes if $c_x \in L_{\texttt{Bit}}$ 
    \begin{equation}
        L_{\texttt{Bit}} = \{c_x: x \in \bit \land c_x = \Com(x, r_x) \}    
    \end{equation}
    where for some $r_x \in \Z_q$. See Appendix \ref{app:sigma_open} for a concrete construction of the $\Sigma$-OR proof using Pedersen Commitment schemes proposed by \cite{damgaard2000efficient}.
    }
\end{enumerate}

In all our experiments and security proofs, we use Pedersen Commitments (PC), though one could replace PC with ~\cite{weng2021wolverine, dittmer2020line, baum2021mathsf}, and still satisfy all the above properties.

\parag{Differential Privacy (DP and IND-CDP)}  We consider two variants of the privacy definition. 

\begin{definition}[Information Theoretic DP~\cite{vadhan2017complexity}]
\label{def:dp_definitions}
Fix $n \in \N, \epsilon \geq 0$ and $\delta \leq n^{-\omega(1)}$. An algorithm $\mathcal{M}: \mathcal{X}^n \times Q \rightarrow \mathcal{Y}$ satisfies $(\epsilon, \delta)$ differential privacy if for every two neighboring datasets $X \sim X^{\prime}$ such that $||X \sim X^{\prime} ||_1 = 1$ and  for every query $Q \in \mathcal{Q}$ we have for all $T \subseteq \mathcal{Y}$
\begin{align}
\label{eq:dp}
    \prob{M(X, Q) \in T} &\leq e^{\epsilon}\prob{M(X^\prime, Q) \in T} + \delta
\end{align}
\end{definition}

A direct corollary of the above definition is that, given $M(X, Q)$ and $M(X^\prime, Q)$, with probability $1- \delta$ even an unbounded Turing Machine $D$ is unable to distinguish between the outputs up to statistical distance $\epsilon$. We will often write $M(X, Q) \stackrel{(\epsilon, \delta)}{\equiv} M(X^\prime, Q)$ as short hand to say that $\mathcal{M}$ is DP.


\begin{definition}[Computational DP~\cite{mironov2009computational}]
\label{def:comp_dp_definitions} 
Fix $\kappa \in \N$ and $n \in \N$. Let $\epsilon \geq 0$ and $\delta(\kappa) \leq \kappa^{-\omega(1)}$ be a negligible function, and let $\mathcal{M} = \{ \mathcal{M}_\kappa : \mathcal{X}_\kappa^n \rightarrow  \mathcal{Y}_\kappa \}_{\kappa \in \N}$ be a family of randomised algorithms, where $\mathcal{X}_\kappa$ and $\mathcal{Y}_\kappa$ can be represented by $\texttt{poly}(\kappa)$-bit strings. We say that $\mathcal{M}$ is \textit{computationally $\epsilon$-differentially private} if for every non-uniform PPT TM's (``distinguishers'') $D$, for every query $Q \in \mathcal{Q}$, and for every neighbouring dataset $X \sim X^\prime$, $\forall T \subseteq \mathcal{Y}_\kappa$ we have 
\begin{equation}    
\Pr\Big[ D(\mathcal{M}(X, q) \in T) = 1\Big] \leq e^{\epsilon}\cdot\Pr\Big[ D(\mathcal{M}(X^\prime, q)\in T) = 1\Big] + \delta(\kappa)
\end{equation}

We will often write $M(X, Q) \stackrel{(\epsilon, \delta)-CDP}{\equiv} M(X^\prime, Q)$ as short hand to say that $\mathcal{M}$ is IND-CDP.
\end{definition}


\begin{definition}(DP-Error) Let $\mathcal{M}: \mathcal{X} \times \mathcal{Q} \rightarrow \mathcal{Y}$ be a $(\epsilon, \delta)$-DP mechanism over $\mathcal{Q}$. Assume that the $L_1$ norm is well-defined on $\mathcal{Y}$. 
For any $n\in \N$, $X \in \mathcal{X}^n$, we define the expected error of the mechanism $\mathcal{M}$ relative to $Q$ as 
\begin{equation}
\label{eq:error}
\texttt{Err}_{\mathcal{M}, Q} = \mathbb{E}[\|Q(X) - \mathcal{M}(X, Q)\|]   
\end{equation}
\noindent
where the expectation is taken over internal randomness of $\mathcal{M}$. \end{definition}

When the context is abundantly clear, to simply notation we will drop the subscripts and refer to equation \eqref{eq:error} as just \texttt{Err}.
It is well known that for negligible $\delta$, the counting query (i.e., DP histograms) has error $\texttt{Err} = O(\frac{1}{\epsilon})$ in the trusted curator model and MPC model~\cite{vadhan2017complexity, corrigan-gibbs_prio_2017}.

\parag{Binomial Mechanism} We use Binomial noise to achieve privacy. 

\begin{lemma}[Binomial Mechanism]
\label{theorem:dp_guarantee}
 Let $X=(x_1, \dots, x_n) \in \Z_q^n$ and define counting query $Q(X) = \sum_{i=1}^n x_i$. Fix $\noisen > 30$, $0 < \delta \leq o(\frac{1}{\noisen})$ and let  $Z \sim \texttt{Binomial}(\noisen, \frac{1}{2})$. Then $Z + Q(X)$ is an $(\epsilon, \delta)$-differentially private with $\epsilon = 10\sqrt{\frac{1}{\noisen}\ln\frac{2}{\delta}}$.
\end{lemma}

It is easy to see that the binomial mechanism incurs constant DP-error (i.e., it is independent of $n$, and depends only on $\epsilon, \delta$). The proof for Lemma \ref{theorem:dp_guarantee} can be found in \cite{ghazi_power_2020}, which we re-derive in Appendix \ref{app:dp_proof_bin_mech} for completeness.

\input{tex_files/morra.tex}

%% file: tex_files/morra.tex
\begin{algorithm}[t]
\caption{$\Pi_\text{morra}$ A protocol for sampling a public coin}
\label{alg: morra}
\flushleft \textbf{Input}: $\lambda_1, \dots, \lambda_K$ \\
\textbf{Output}: $z \xleftarrow{R} \bit$ 

\begin{enumerate}
    \item {Each server $k \in [K]$ is asked to sample $m_k \xleftarrow[]{R} \Z_q$ uniformly at random.}
    
    \item{\textit{Commit}: Each server samples $r_{m_k} \xleftarrow[]{R} \Z_q$ and broadcasts $c_{k} = \Com(m_k,r_k)$ to all other servers. Assume without loss of generality that the servers broadcast their commitments in natural lexicographical order $k \in [K]$. }
    
    \item{\textit{Reveal}: Once all servers have received $c_k$, they now broadcast $m_k, r_{m_k}$ to all servers in the reverse order in which the commitments arrived. 
    It is important that the reverse order is respected as it guarantees that each server's inputs are independent of the inputs of other servers. 
    Once all commitments are revealed, each server verifies that $\Com(m_k, r_k) = c_k$. If this test fails for any $k$ or one of the servers does not respond, the protocol is aborted.}
    
    \item{Each server computes $X = (m_1 + \dots + m_k) \mod q$. We have $X \xleftarrow{R} \Z_q$. If $X \leq \lceil \frac{q}{2} \rceil$ then $c_i = 0$. Otherwise $c_i = 1$. Thus we can use this protocol to generate unbiased coins and uniformly random values. }
\end{enumerate}
\end{algorithm}

\parag{Morra}
\label{sec: morra}
We will prove zero knowledge (or security for MPC) assuming that the provers and verifiers (or all participants of the MPC, respectively) have access to an oracle that returns a polynomial sized stream of publicly random unbiased bits. 
In other words, we assume that all parties have access to an oracle functionality $\oracle_\texttt{morra}(1^{\kappa}, \lambda_1, \dots, \lambda_K) = z$ where $z \xleftarrow{R} \bit$ where and $\lambda_k$ refers to the empty string for all $k \in [K]$. 

In practice, this oracle is replaced by a lightweight MPC protocol such as $\Pi_{\texttt{morra}}$ defined in Algorithm~\ref{alg: morra}, which is a modification of an ancient game called Morra\footnote{\url{https://en.wikipedia.org/wiki/Morra_(game)}}, that securely computes $\oracle_\texttt{morra}$ in the presence of a dishonest majority of active participants.
It is easy to see that as long as one participant is honest and samples its value uniformly at random, the final protocol produces an unbiased coin.  Since the commitment is hiding, a corrupt party cannot infer any information about the other parties choice $m_k$ from the published $c_{m_k}$ and by the binding property, a participant cannot change their decision after observing another party's opening. A formal simulator-styled proof can be found in Blum's seminal work for flipping coins over a telephone~\cite{blum1983coin} or any introductory textbook on MPC (under the title weak coin flipping).  If we omit the final thresholding step, the above protocol can be used to sample $z \xleftarrow{R} Z_q$.

%% file: tex_files/new_definitions.tex
\section{Security Models for Verifiable DP}
\label{sec:model}

This section introduces verifiable DP in both the single trusted curator and MPC model. 
In both settings, the input comes from $n$ distinct clients. 
Informally, the main difference between the two models is that the former has plaintext access to the client data. In contrast, in MPC-DP, the clients secret share (or partition) their inputs and each server receives information theoretically hiding shares (or a partial view) of client inputs. Additionally, instead of a single trusted entity computing $\mathcal{M}$, the servers participate in an MPC protocol $\Pi$ to securely compute $\mathcal{M}$ without revealing any information other than $\mathcal{M}(x_1, \dots, x_n, Q)$. 

 For some queries $Q \in \mathcal{Q}$, the protocol requires that the client inputs come from a restricted subset $L \subseteq \mathcal{X}$. 
 For such cases, if the servers are operating on information-theoretically hiding shares of the inputs, the clients must send a zero-knowledge proof so that the provers can verify that the inputs come from the specified language, without learning any other information about the inputs.
 Examples of such proofs can be found in the prior literature~\cite{boneh_lightweight_2022, boyle2019secure, corrigan-gibbs_prio_2017, bunz2018bulletproofs}. 
 In the definitions below and in what follows, we use the terms $\Pv$ (prover), server and curator interchangeably and the terms analyst and $\Vfr$ (verifier) to refer to the same entity.

\input{tex_files/single_curator_dp_defs.tex}

%% file: tex_files/single_curator_dp_defs.tex
\subsection{Verifiable DP}

In what follows, we describe the MPC model and then discuss how it can be specialized to the trusted curator model. 
Let $\mathcal{M}$ be a DP (or IND-CDP) mechanism as described in Definition \ref{def:dp_definitions} (or Definition \ref{def:comp_dp_definitions} resepectively) for a query class $\mathcal{Q}$. 
Let $\kappa \in \N$ denote the security parameter. 
A verifiable DP mechanism for $\mathcal{M}$ for query class $\mathcal{Q}$ on a dataset $X$, consists of two interactive protocols $\texttt{Setup}$ and $\Pi$.
The first protocol denoted by $\texttt{Setup}$ is used to generate public parameters. Let $\pp \leftarrow \texttt{Setup}(1^\kappa)$ denote such public parameters. 
The second protocol $\Pi$ is a multiparty protocol  between $K+1$ ``next-message-computing-algorithms'' algorithms $\Vfr$ and $(\Pv_1, \dots, \Pv_K)$. 
The total number of rounds of message passing is upper bounded by some polynomial $\texttt{poly}(\kappa)$. 
In message passing algorithms, $\Vfr$'s (respectively $\Pv_k$'s)  message $m_i$ at round $i$ is determined by its input, messages it has received so far from $\Pv_k$ (respectively $\Vfr$) and internal randomness $\bvec{r}_v$ (respectively $\bvec{r}_{\Pv_k}$).
Let $\vec{\Pv}$ denote a succinct representation for $(\Pv_1, \dots, \Pv_K)$.
Each prover $\Pv_k$ receives on its input tape n inputs $\Big( \share{x_1}_k, \dots, \share{x_n}_k\Big)$, succinctly denoted by $\vec{X_k}$.
  Let $\bvec{r}_{\Pv_k} \in \bit^{\texttt{poly}(\kappa)}$ denote the internal randomness for $\Pv_k$;
let $z \in \bit^{\texttt{poly}(\kappa)}$ denote auxiliary input available to the verifier; and let $\bvec{r}_v \in \bit^{\texttt{poly}(\kappa)}$ denote the internal randomness used by the verifying algorithm. 
At the end of the protocol, the provers send $\bvec{y} \in \mathcal{Y}$ to the $\Vfr$, who then outputs either 0 or 1, with 1 indicating that the verifier accepts the provers' claim that, the real protocol output is indistinguishable from an ideal computation, i.e. $\bvec{y}  = \mathcal{M}(X, Q)$. 
Let $\texttt{out}(\Vfr, \bvec{y}, \bvec{r}_v,  \bvec{r}_{\vec{\Pv}}, z, \vec{\Pv}, \pp) \in \bit$ denote the verifying algorithm's decision.  
In the definition below, we write $\texttt{out}(\Vfr, \Pv)$ as shorthand for $\texttt{out}(\Vfr, \bvec{y}, \bvec{r}_v,  \bvec{r}_{\vec{\Pv}}, z, \vec{\Pv}, \pp)$.
The trusted curator can be understood as essentially this model with a single prover, i.e., we set $K=1$. 
Thus the only functional difference between MPC-DP and trusted curator DP is that in the latter case, the curator sees all the data in plaintext. In MPC, the data may be secret, shared or partitioned across the provers. In both cases, the prover(s) must prove that they did not tamper with the protocol to generate an output that is distinguishable from the ideal computation of $\mathcal{M}$.

\begin{definition}[Verifiable DP]\label{defn:vdp_MPC}
 A constant round interactive verifiable DP protocol for $\mathcal{M}$ consists of two algorithms $\texttt{Setup}$ and $\Pi$, such that for $n \in \N$ clients, $K \geq 1$ provers denoted by $\vec{\Pv}$ and a single verifier $\Vfr$, there exists negligible functions $\delta_c$ and $\delta_s$ such that 

\begin{enumerate}
    \item{\textbf{Completeness:} Let $X=(x_1, \dots, x_n) \in \mathcal{X}^n$ be the client inputs that have been split in $K$ shares $(\vec{X_1}, \dots, \vec{X_K})$, where $\vec{X_j}$ denotes the input sent to the $j$'th prover, then as long as the $\vec{\Pv}$ and $\Vfr$ honestly execute $\Pi$, then we have
    
     \[ \Pr\left[ \texttt{out}(\Vfr, \vec{\Pv}) = 0 : \begin{array}{c} \pp \leftarrow \texttt{Setup}(1^\kappa) \\
      \Pv_j \leftarrow (\vec{X_j}, \bvec{r}_{\Pv_j}, \pp)\\
     \Vfr \leftarrow (z, \bvec{r}_v, \pp) \\
\bvec{y} \leftarrow \Pi(\vec{\Pv} ,\Vfr, \pp)
    \end{array} \right]
  \leq \delta_c . \] 
  }

    \item{\textbf{Soundness:} For every $X \in \mathcal{X}^n$ and any subset $I \subseteq [K]$, let $\vec{\Pv^*}$ denote the collection of corrupted provers, indexed by $I$, that deviate from $\Pi$, such that the final output $\bvec{y} \neq \mathcal{M}(X, Q)$, then we have

     \[ \Pr\left[ \texttt{out}(\Vfr, \vec{\Pv}) = 1 : \begin{array}{c} \pp \leftarrow \texttt{Setup}(1^\kappa) \\
           \Pv_j \leftarrow (\vec{X_j}, \bvec{r}_{\Pv_j}, \pp) \\
     \Vfr \leftarrow (z, \bvec{r}_v, \pp) \\
\bvec{y} \leftarrow \Pi(\vec{\Pv^*} ,\Vfr, \pp)
    \end{array} \right]
  \leq \delta_s . \] 
  
  Note that the correctness of the protocol is defined in terms of the actual inputs the clients sent to $\mathcal{M}$ and not the inputs a corrupted set of provers might have used. Thus if corrupted prover(s) tampered with the inputs on its input tape and then executed the protocol faithfully, by our definitions, it still counts as cheating.

}
     
    \item{\textbf{Zero-Knowledge}: Let $\vec{\Pv^*}$ denote the collection of provers corrupted by any verifying algorithm $\Vfr^*$ and $I \subset [K]$ denote their indices. Let $\vec{\Pv^*}$ denote the set of honest provers not indexed by $I$.
    Let $\texttt{View}\Big[\Pi\Big((\vec{\Pv}, \vec{\Pv^*}) ,\Vfr^*\Big)\Big]$ be the joint distribution\footnote{As $\mathcal{M}$ is a random function, the \textit{joint distribution} of the view of the adversary and their output must be indistinguishable from the simulated transcript (and not just the view of the adversary). See \cite{lindell2017simulate} for more details.} of messages and output received during the execution of $\Pi$ in the presence of corrupted parties. There exists a PPT Simulator $\Sim_{(\Vfr^*, I)}$ such that for all $\bvec{y} = \mathcal{M}(X, Q)$
    \begin{align*}
    \texttt{View}\Big[\Pi\Big((\vec{\Pv}, \vec{\Pv^*}) ,\Vfr^*\Big)\Big] &\stackrel{}{\equiv} \Sim_{(\Vfr^*, I)}(\bvec{y}, \bvec{r}_v, z, \pp)        
    \end{align*}

    where $z \in \bit^{\texttt{poly}(\kappa)}$ represents auxiliary input available to all the corrupted parties. Contrary to soundness, for zero-knowledge to hold, the simulated transcript should be indistinguishable from the actual protocol transcript, based on the inputs adversaries used and not the ones the clients sent to a set of corrupted provers.
    }

\end{enumerate}
 \end{definition}
 
An interesting point to note is that in verifiable differential privacy, the verifier plays a dual role. An honest verifier ensures that the output is faithfully generated and thus plays an active role in generating the DP noise without ever seeing it in plaintext. On the other hand, a dishonest verifier tries to tamper with the protocol to breach privacy. In non-verifiable DP, the analysts (verifier) only have access to DP the statistic. They have no agency over how the output is generated. Thus the verifier participating in verifiable DP has a greater attack surface than a classical adversary in traditional non-verifiable DP. We elaborate on this in Section \ref{sec:separation}, when trying to establish separations between statistical DP and computational DP. Additionally, just like in standard MPC, in the presence of a dishonest majority of corrupted participants, we do not treat early exiting by corrupted parties as a breach of security. This is easily detected by the honest parties, and the output is ignored. Verifiable DP, just like interactive zero-knowledge proofs \cite{goldreich_foundations_2007} comes in 24 different flavours based on the capabilities of the corrupted parties:
 
 \begin{enumerate}
 	\item{\textbf{Distinguishability:} Based on the distinguishability properties of the simulator algorithm, the protocol may be perfect, statistical or computationally zero-knowledge. The protocol described in Section \ref{sec:single_curator_hist} is computationally zero-knowledge. } 
 	\item{\textbf{Verifier specifications:} Based on whether the verifier is expected to follow the rules of the protocol (semi-honest) or may deviate arbitrarily (active), we get honest-verifier zero knowledge or just zero knowledge. All our results are zero-knowledge.}  	
 	\item{\textbf{Soundness:} Based on the power of the corrupted provers, the protocol may be computationally sound (also known as arguments) or statistically sound (secure against unbounded provers). The verifiable DP protocol in Section \ref{sec:single_curator_hist} is computationally sound.}  	
 	\item{\textbf{Inputs:} Based on whether the verifier has access to the auxiliary input, the protocol could be plaintext zero-knowledge or auxiliary input zero knowledge. Our protocols allow for the verifier to have auxiliary input.}  	 	
 \end{enumerate}

%% file: tex_files/single_curator_bin_mechanism.tex
\section{Verifiable Binomial Mechanism}
\label{sec:single_curator_hist}

This section describes how to compute counting queries verifiably with differential privacy in both the single curator and client-server MPC models. 
We consider  the trusted curator model to be a special instantiation of the general MPC model where the number of provers $K=1$.
In Section \ref{sec:intuition} we describe intuitions for our protocol, and in Section \ref{sec:client-server-mpc-dp} we explain what is needed for verifiability in the MPC setting, and tackle the additional challenges of verifying client inputs.
We describe how prior efforts at verifying clients fall short of the security expectations of Definition~\ref{defn:vdp_MPC}. 
Finally, in Section \ref{sec:main_protocol} we describe a protocol that computes counting queries with DP verifiably.

Set $\mathcal{X} = \Z_q = \mathcal{Y}$, where $\Z_q$ is a prime order finite field of size $q$ over the integers. Let $X=(x_1, \dots, x_n)$ denote the client inputs and the $Q$ be the counting query $Q(X) = \sum_{i=1}^n x_i$.  Let $\share{x_i}_k$ denote the $k$'th additive secret\footnote{Although we describe our protocols with additive secret sharing, any linear secret sharing such as Shamir's secret sharing also applies to all our results.} share of a client input $x_i$. Each client splits their input into $K$ secret shares and distributes them across the provers. We will assume that $n\ll q$ and $\kappa= \lfloor \log_2 q\rfloor$ can be viewed as the security parameter. For $K \geq 1$ provers and 1 verifier, define the oracle functionality $\mathcal{M}_{\texttt{Bin}}$ in the ideal world as follows:
\begin{enumerate}
	\item{$\mathcal{M}_{\bin}$ receives public privacy parameters $\epsilon$ and $\delta$. It then computes $\noisen$ based on Lemma \ref{theorem:dp_guarantee}.}

	\item{Let $\Big(\share{x_1}_k, \dots,\share{x_n}_k\Big)$ denote the inputs on the $k$'th prover's input tape. Each prover $\Pv_k$  is expected to compute $X_k = \sum_{i=1}^n \llbracket x_i \rrbracket_k$ and sends to $\mathcal{M}_{\bin}$ as its input $X_k$. A corrupted prover might send an arbitrary input.} 
	
\item{$\mathcal{M}_{\bin}$ samples $\Delta_k \sim \texttt{Binomial}(\noisen, 1/2)$ independently for each input $X_k$ it receives. It then computes 
\begin{equation}
\label{eq:M_bin}
\textstyle
y =  \sum_{k=1}^K (X_k + \Delta_k)
\end{equation}
}
\item{$\mathcal{M}_{\bin}$ sends the tuple $(y, \Delta_k)$ as output to each prover $\Pv_k$. On receiving its output, the $\Pv_k$ sends $\texttt{CONTINUE}$ to $\mathcal{M}_{\bin}$. Once $\mathcal{M}_{\bin}$ receives the continue signal from prover $\Pv_k$ it moves on to deliver output to $\Pv_{k+1}$.}

\item{After all $K$ provers have sent $\texttt{CONTINUE}$, $\mathcal{M}_{\bin}$ sends $y$ as output to the verifier $\Vfr$. If a single prover fails to send the continue message and thereby exits the protocol early, the verifier and the remaining provers do not receive any output.}
\end{enumerate}

When $K=1$, i.e., the trusted curator setting, the single prover receives $n$ client inputs in plaintext, so $\share{x_i}_k = x_i$ for all $i \in [n]$. 
This is equivalent to an adversary corrupting all $K$ provers. Thus in the MPC setting with $K \geq 2$ servers, it is safe to assume at least one of them will follow the protocol.
Our goal is to be able to come up with an interactive protocol $\Pi_\bin$, 
which allows us to compute $\mathcal{M}_{\texttt{Bin}}$ verifiably as per Definition~\ref{defn:vdp_MPC}. 
Notice that in the ideal model definition above, the oracle adds $K$ independent copies of DP noise to the output, whereas 
 Lemma \ref{theorem:dp_guarantee} only calls for a single copy. 
 This is because, as we allow up to $K-1$ provers to collude with a corrupted verifier, the corrupted provers could simply not add any noise to the output. 
Ben Or \etal's completeness results  \cite{ben2019completeness} imply that $K$ independent copies of noise are \textit{necessary} to guarantee differential privacy unless the number of corruptions can be restricted to being strictly less than $\frac{K}{3}$, so each prover must independently generate enough noise to guarantee DP.
Our protocols defined below are secure against computationally bounded provers and verifiers that may deviate arbitrarily from protocol specifications and have access to auxiliary inputs.

\input{tex_files/protocol.tex}

\begin{theorem}
\label{theorem:verifiable_dp_feasible} 
Let $X=(x_1, \dots, x_n)$ be the client input. Let $\mathcal{M}_{\texttt{Bin}}$ and $\oracle = (\oracle_{\morra}, \oracle_{\texttt{OR}})$ be as defined above. 
Assuming $\Pi_{\texttt{Bin}}$ is run with $K \geq 1$ provers and a single verifier, then the following is true

    \item{\textbf{Completeness:} For every $X \in \mathcal{X}^n$ 
     \[ \Pr\left[ \texttt{out}(\Vfr, \vec{\Pv}) = 0 : \begin{array}{c} \pp \leftarrow \texttt{Setup}^\oracle(1^\kappa) \\
     \Pv_k^\oracle \leftarrow \share{X}_k, \bvec{r}_{\Pv_k}, \pp \\
     \Vfr^\oracle \leftarrow z, \bvec{r}_v, \pp \\
y \leftarrow \Pi_\bin^\oracle(\vec{\Pv}, \Vfr)
    \end{array} \right]
  = 0  \] 
  
  where $\share{X}_k = (\share{x_1}_k, \dots, \share{x_n}_k)$ and $\bvec{r}_{\Pv_k}$ denotes $\Pv_k$'s private randomness.
  }

\item{\textbf{Computational Soundness:} For every $X \in \mathcal{X}^n$ and any subset $I \subseteq [K]$, let $\vec{\Pv^*}$ denote the collection of provers, indexed by $I$, that have been corrupted by an adversary $\AdvA$, such that the final output $y \neq \mathcal{M}_\bin(X, Q)$. Let $\vec{\Pv}$ denote the collection of honest provers not indexed by $I$. Let $z$ denote the auxiliary input available to $\AdvA$ and $\mu$ be its advantage in the discrete log game (Definition \ref{defn:discrete_log}) 

  \[ \Pr\left[ \texttt{out}(\Vfr, \vec{\Pv^*}, \vec{\Pv}) = 1 : \begin{array}{c} \pp \leftarrow \texttt{Setup}^\oracle(1^\kappa) \\
     \Pv_k^\oracle \leftarrow \llbracket X \rrbracket_k, \bvec{r}_{\Pv_k}, \pp \\
     \Vfr^\oracle \leftarrow z, \bvec{r}_v, \pp \\
y \leftarrow \Pi_\bin^\oracle\Big((\vec{\Pv}, \vec{\Pv^*}), \Vfr\Big)
    \end{array} \right]
  \leq \mu(\kappa)  \] 
  
Note that as $I \subseteq [K]$, soundness, as defined above covers both the MPC and the trusted curator setting.
}
     
    \item{\textbf{Computational Zero-Knowledge}: 
Let $\vec{\Pv^*}$ denote the collection of provers, indexed by $I \subset [K]$, that have been corrupted by a corrupt verifier $\Vfr^*$. There exists a PPT Simulator $\Sim_{(\Vfr^*, I)}$ such that for all $y = \mathcal{M}(X, Q)$
    \begin{align*}
    \texttt{View}\Bigg[\Pi\Big((\vec{\Pv}, \vec{\Pv^*}), \Vfr^*, \pp\Big)\Bigg] &\stackrel{c}{\equiv} \Sim_{(\Vfr^*, I)}(y, \bvec{r}_v, z, \pp)        
    \end{align*}

    where $z \in \bit^{\texttt{poly}(\kappa)}$ and $\bvec{r}_v \in \bit^{\texttt{poly}(\kappa)}$ represents auxiliary input and randomness available to all the corrupted parties.
    }

\end{theorem}

\begin{proof}
\
\begin{enumerate}
\input{tex_files/completeness.tex}
\input{tex_files/audit.tex}
\input{tex_files/zero_knowledge.tex}\end{enumerate}\end{proof}

%% file: tex_files/protocol.tex
\subsection{An Intuitive But Incomplete Protocol}
\label{sec:intuition}
Before describing the full protocol in Section \ref{sec:main_protocol} and Figure \ref{fig:dp_interactive_proof}, we provide the reader with some intuition as to why the protocol works for a single curator and verifier. \textit{In this section, we make the unrealistic assumption that prover and verifier behave faithfully}. 
Assume all parties have joint oracle access to $\oracle_{\texttt{Morra}}$ (as described in Section \ref{sec: morra}) to jointly sample unbiased bits $(b_1, \dots, b_{\noisen})$. It is easy to see that using $(\sum_{i=1}^{\noisen} b_i)$ as DP randomness results in the desired distribution defined in $\mathcal{M}_{\texttt{Bin}}$. However, the oracle output is publicly known to both the verifier and prover; therefore, it cannot be directly used to guarantee differential privacy. 
As discussed earlier, this problem of proving that a prover faithfully sampled random bits without disclosing them lies at the heart of any verifiable DP protocol. 
Thus the protocol must combine public coins that satisfy verifiability requirements and private coins that ensure secrecy.

The protocol for verifiable DP counting proceeds in $\noisen$ identical and independent invocations (run in parallel). 
In copy $i$, the prover samples $v_i \in \bit$, which it keeps private. 
Note that a prover could sample this bit using any arbitrary bias. 
As this is the provers' private coin, the verifier has no control over how the prover generates this information. 
After the prover has sampled their private bit, the prover and verifier make one call to $\oracle_{\texttt{Morra}}$ to get an unbiased coin denoted by $b_i$. 
Next, the prover locally computes  $\hat{v}_i = b_i \oplus v_i$. Here $\oplus$ refers to the boolean XOR operation. It is easy to see that $\hat{v}_i$ has the same distribution as $b_i$, but its value is known only to the parties with access to $v_i$, i.e., the prover. 
After $\noisen$ rounds, the prover computes $Q(X)$ and $Z = \sum_{i=1}^{\noisen} \hat{v}_i$ and outputs $Q(X) + Z$ where $Z$ is used as DP randomness.  By the assumption that the prover and verifier are faithful, $Z$ is distributed according to the desired distribution stated in Theorem \ref{theorem:dp_guarantee}, and its value is only known to the prover. 
To make this protocol practical, we need to resolve a few issues.

\begin{enumerate}
    \item{Although the above description requires a bitwise XOR operation to ensure the right distribution is used, we operate with arithmetic circuits in the actual protocol. 
    Thus, the provers could sample arbitrary values $v^* \in \Z_q$ such that $v^* \notin \bit$, and we need to fix how to express the XOR operation via arithmetic circuits. 
    }

    \item{Even if we could verify that the prover sampled a private bit correctly, we still need to verify that they faithfully performed the local operations discussed above.}
\end{enumerate}

Thus, if we could guarantee that each server performed its computations correctly and sampled a private value from the correct set, we would get the desired outcome of verifiable and DP counting queries.

\begin{figure}[t]
    \centering    
	\includegraphics[width=0.8\textwidth]{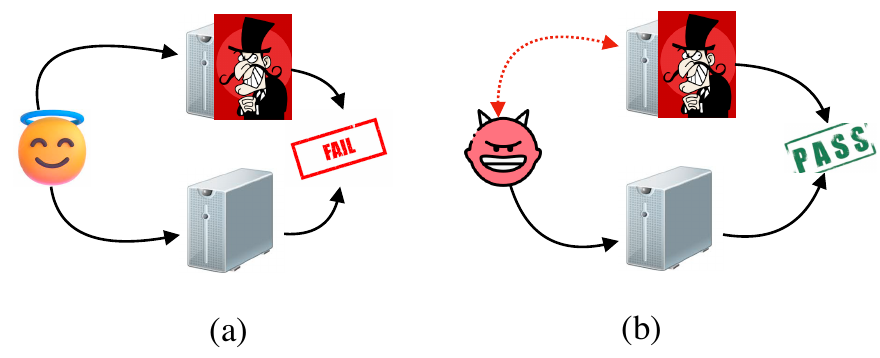}
	\caption{Two types of attacks that go undetected in Poplar. In (a) regardless of what the honest client sends, a corrupted server simply ignores the input and excludes the client from the protocol based on auxiliary information. In (b) a dishonest client colludes with the corrupted server by revealing secret values, so that an illegal input is included. 
 In both cases, the honest server cannot distinguish between an honest run and a corrupted run of the protocol.}
	\label{fig:attacks}
\end{figure}

\subsection{Extending To Client-Server MPC-DP}
\label{sec:client-server-mpc-dp}

To compute DP histograms verifiably in the client-server MPC-DP setting, we use the same computational model used for PRIO \cite{corrigan-gibbs_prio_2017} and Poplar \cite{boneh_lightweight_2022}, two real-world systems deployed at scale by Mozilla\footnote{\url{https://blog.mozilla.org/security/2019/06/06/next-steps-in-privacy-preserving-telemetry-with-prio/}}. 
As discussed earlier, in this setting $n$ clients secret share their inputs $x_i \in L$ amongst $K \geq 2$ provers, where $L \subseteq \mathcal{X}$ defines the language of legal inputs to the protocol. 
For computing $M$-bin histograms over $n$ inputs, $L$ is the set of all one-hot encoded vectors of size $M$. 
For the core problem of a single-dimensional counting query, $M=1$ and $L=\bit$.  
Since the inputs on the prover's tapes reveal no information about a client's input, for the protocol to be useful the provers must first verify in zero-knowledge that $x_i \in L$ before using such inputs to compute aggregate statistics. 
This additional step of verifying a client is not required in the trusted curator model, as the prover decides what inputs should be included in the computation and can see them in plaintext.

\parag{Verifying Clients in MPC-DP}
\label{sec:client_verification}
Poplar and PRIO use efficient sketching techniques from \cite{boyle_function_2016} to validate a client's input in zero knowledge \textit{without} relying on any public key cryptography. Thus, as long as at least one out of $K$ provers does not reveal the inputs it received, even an unbounded adversary corrupting the remaining provers cannot ascertain any information about an honest client's input. While such a system protects an honest client's privacy from an unbounded adversary, it is not verifiable as per Definition \ref{defn:vdp_MPC}. 
Specifically, for the techniques used in PRIO and Poplar, a single corrupted prover could tamper with its inputs and exclude an honest client from the protocol by forcing them to fail the verification test. 
Alternatively, a corrupt client could collude with a prover to include arbitrary inputs, jeopardising the correctness of the output. 
Figure~\ref{fig:attacks} summarises these attacks on Poplar and PRIO\footnote{Concretely, in scenario (b), the dishonest client reveals the values $\kappa$ and $[v]_0$ to the server. This allows the server to set $z_1 = -z_0, z_1^* = -z_0^*$ and $z_1^{**} = -z_0^{**}$, thereby 
 admitting an illegal input into the protocol.}. By our definitions of verifiability, the protocol's output \textit{must} be a function of the inputs provided by honest clients only. 
 Thus the protocol described in Section \ref{sec:main_protocol} provides the following additional guarantees:

\begin{enumerate}
	\item{\textbf{Guaranteed Inclusion Of Honest Clients}: If a client submits shares of an input $x \in L$, then the final output of the protocol is guaranteed to use this input untampered. Thus an honest client is assured that, as long as a single prover follows the protocol specifications, no one learns any information about their private input and their input is correctly used to compute the final output. }
	\item{\textbf{Guaranteed Exclusion Of Corrupt Clients}: A corrupted client, even one that has control over any proper subset of the $K$ provers, cannot include an invalid input to the protocol. Thus if $x \notin L$, $x$ is discarded by our protocol with overwhelming probability.}
\end{enumerate}

It is important to note that as we operate under stricter notions of privacy and correctness, our results require the use of public-key cryptography and security holds only against computationally bounded adversaries.  Furthermore, we show in Section \ref{sec:separation} that it is impossible to satisfy verifiable DP and provide information theoretic guarantees.

\begin{figure*}
    \centering
\begin{pchstack}[boxed]  
    \pseudocode[linenumbering , skipfirstln]{%
    \textbf{Public Verifier}(\Vfr) \< \< \textbf{Prover}(\Pv_k) \\[][\hline]
    \texttt{pp} \leftarrow \texttt{Setup}(1^\kappa) \< \text{Generate public parameters} \<   \texttt{pp} \leftarrow \texttt{Setup}(1^\kappa) \\
	\Bigg\{  \Big\{ c_{i,k}  \Big\}_{k \in [K]} \Bigg\}_{i \in [n]} \< 
 \text{Client inputs }  \<  \Big\{ \share{x_i}_k, r_{i,k} \Big\}_{i \in [n]}\\
  \text{Verify using } \oracle_{\texttt{OR}} \text{ that} \forall i \in [n] \text{, }x_i \in L    \<  \< \text{Verify using } \oracle_{\texttt{OR}} \text{ that} \forall i \in [n] \text{, }x_i \in L   \\
    (c'_{1,k}, \dots, c'_{{\noisen},k})\< \sendmessageleft*{c'_{j,k} = \Com\Big( v_{j,k}, s_{v_{j,k}}\Big)} \< \forall j \in [\noisen] \text{ Samples and commits } v_{j,k} \in \bit \text{}\\    
    \forall j \in [\noisen]\text{ Send } c'_{j,k} \< 
\oracle_{\texttt{OR}}  \<  \forall j \in [\noisen] \text{ Send openings} (v_{j,k}, s_{j,k})\\    
    \forall j \in [\noisen]\text{ Check } \oracle_{\texttt{OR}}(c'_{j,k}) = 1\<\<  \\
    \forall j \in [\noisen]\text{ Send empty string } \lambda_j \< 
\oracle_{\texttt{Morra}} \< \forall j \in [\noisen]\text{ Send empty string } \lambda_j\\
    \text{ Receive } (b_{1,k}, \dots, b_{\noisen,k}) \< 
 \forall j \in [\noisen] \text{ }b_{j,k} = \oracle_{\texttt{Morra}}(\lambda_j) \text{ } \< \text{ Receive } (b_{1,k}, \dots, b_{\noisen,k})\\  \< \< \forall j \in [\noisen] \text{ Adjust } v_{j,k} \text{ based on }  b_{j,k}, \text{ to get } \hat{v}_{j,k}\\ 
  \< \sendmessageleft*{y_k} \< y_k = \sum_{i=1}^n \share{x_i}_k + \sum_{j=1}^{\noisen} \hat{v}_{j,k} \\
  \< \sendmessageleft*{z_{k}} \<  z_{k} = \Big(\sum_{i=1}^{n}r_{i,k}+ \sum_{j=1}^{\noisen} s_{j,k}\Big) \\  
    \text{Compute } \hat{c}'_{j,k} \text{ using } b_{j,k}  \text{ for all } j \in [b_{\noisen}]\< \<  \\    
  \text{Check that } \Big( \prod_{i=1}^{n} c_{i,k} \times \prod_{j=1}^{\noisen} \hat{c}'_{j,k}\Big) = \Com(y_k, z_{k}) \< \<
}
\end{pchstack}    
    \caption{The figure above describes the interaction between a single prover and verifier in $\Pi_{\texttt{Bin}}$. In the single trusted curator model $K=1$ we have $x_i = \share{x_i}_k$ where the prover can see client inputs in plaintext. In the MPC setting, each prover $\Pv_k$ follows the exact same protocol on their respective inputs specified in Line 2. Thus at the end of the protocol, each prover $\Pv_k$ outputs the tuple $y_k, z_{k}$. A public verifier aggregates the output from each prover to publish verifiable DP statistics. }
    \label{fig:dp_interactive_proof}
\end{figure*}

\subsection{Main Protocol Description}
\label{sec:main_protocol}

The protocol $\Pi_{\bin}$ described in Figure \ref{fig:dp_interactive_proof} provides a compact standalone description of the interaction between $K$ provers and a public verifier for computing $\mathcal{M}_\bin$. As the verifier is public, anyone (even non-participants to $\Pi_\bin$) can see the messages it receives from the clients. 
We assume that both the provers and the verifier have access to oracles $\oracle_{\morra}$ and  $\oracle_{\texttt{OR}}$ as defined in Section \ref{sec: commitments}. In the real world, $\oracle_{\texttt{Morra}}$ is replaced with $\Pi_{\texttt{Morra}}$ (see Algorithm \ref{alg: morra}) and $\oracle_{\texttt{OR}}$ is replaced by Cramer \etal's $\Sigma$-OR proof \cite{cramer1994proofs} (see Appendix \ref{app:sigma_open} for an example implementation) which securely compute the oracle functionalities in the presence of adversaries that may deviate from protocol specifications. Thus, we define our protocol in the hybrid world, and by the sequential composition theorem\footnote{Though we use sequential composition, both protocols $\Pi_{\morra}$ and $\Pi_{\texttt{or}}$ can be parallelly composed.} \cite{goldreich_foundations_2007}, the security properties of the protocol are preserved. 
Next, we describe the protocol in detail with line references to Figure \ref{fig:dp_interactive_proof}:

\begin{enumerate}
    \item[Line 1:]{ In the first step, the prover(s) and verifier agree upon the public parameters for the protocol. The public parameters include a description of $\mathcal{C}_\pp = \G_q, \mathcal{M}_\pp = \mathcal{X} = \mathcal{Y} = \Z_q, \mathcal{R}_\pp= \Z_q$ and a description of $\mathcal{M}_\texttt{bin}$ as defined in equation \eqref{eq:M_bin}. 
    The group $\G_q$ satisfies the requirements of the homomorphic commitment scheme defined in Section \ref{defn:hom_coms}.}

\item[Line 2:]{ For each client $i \in [n]$, let $\share{x_i}_k$ denote the $k$'th share of their input $x_i \in L$. 
Define $c_{i,k} = \Com\Big(\share{x_i}_k,  r_{i,k}\Big)$ as the commitment to the $k$'th share of $x_i$. 
The client sends to each prover $\Pv_k$ the tuple $(\share{x_i}_k,  r_{i,k})$ and publicly broadcasts 
the commitments to their shares $\Big( c_{i,1}, \dots, c_{i,K} \Big)$ to the public verifier observable to all parties.}

\item[Line 3:]{Similar to PRIO and Poplar, we use $L=\bit$, and thus verifier and the client use the oracle $\oracle_{\texttt{OR}}$ to check if the client's input is indeed a commitment to a bit. 
For input $x_i$, the verifier sends to  $\oracle_{\texttt{OR}}$ the derived commitment $c_{i} = \prod_{k=1}^K c_{i,k}$ and the client sends the openings $\Big(x_i, \sum_{k=1}^K r_{i,k}\Big)$. 
The oracle responds with $\oracle_{\texttt{OR}}(c_{i}) = 1$ if $x_i \in \bit$ and $c_i$ is a commitment to $x_i$.
Note that all messages sent to the verifier are public, so the servers can independently validate the verifier's claims. As there is always one honest participant in the protocol (at least one prover or the verifier), this step provides a public record of honest and dishonest clients. 
Such a record resolves the issues presented in Figure~\ref{fig:attacks}. From here on, the protocol only uses inputs from validated clients. }

    \item[Line 4:]{$\Pv_k$ samples $(v_{1,k}, \dots, v_{\noisen,k})$ where $v_{j,k} \in \bit$ (private random bit) and sends to the verifier commitments to $v_{j,k}$ for $j \in [\noisen]$. 
    Let $c'_{j,k} = \Com(v_{j,k}, s_{j,k})$ denote the commitment to $v_{j,k}$ with randomness $s_{j,k}$. To enforce consistency in notation and improve readability, we always use $c$ to denote commitments to client inputs and $c'$ to denote commitments to the prover's private inputs. Similarly, we will always use $r$ and $s$ to denote the randomness used for client input and prover bit commitments, respectively.}    

    \item[Line 5-6:] {The verifier uses $\oracle_{\texttt{OR}}$ to check if the messages sent by the prover were indeed commitments to 0 or 1. 
    This step is essential for the boolean to the arithmetic conversion, as the linearisation of the XOR operation is only valid for values $v \in \bit$ (see completeness property of Theorem~\ref{theorem:verifiable_dp_feasible}).}

    \item[Line 7-8:]{If for any $i \in \noisen$, $\oracle_{\texttt{OR}} = 0$, the verifier aborts the protocol and publicly declares that $\Pv_k$ cheated. Otherwise, once all commitments are verified, the prover and verifier jointly invoke $\oracle_{\morra}$ to get $\noisen$ \textit{public} unbiased bits $(b_{1,k}, \dots, b_{\noisen,k})$.}

    \item[Line 9:]{For all $i \in [\noisen]$, based on the value of $b_{j,k}$, the prover sets $\hat{v}_{j,k}$ as follows
    \[   
    \hat{v}_{j,k} = 
         \begin{cases}
        1 - v_{j,k} & \text{if } b_{j,k}=1\\
        v_{j,k}& \text{otherwise.} \\ 
     \end{cases}
\]
As long as $v_{j,k} \in \bit$, the above set of equations is equivalent to setting $\hat{v}_{j,k} = v_{j,k} \oplus b_{j,k}$. An important feature of this step is that, conditioned on $b_{j,k}$, the operations described above are linear. In Line 11, we describe why this is critical for correctness to hold.
}

    \item[Line 10-11:]{The prover sends $(y_k, z_{k})$ to the verifier:
    \begin{equation}
    y_k = \Big(\sum_{i=1}^n  \share{x_i}_k + \sum_{j=1}^{\noisen} \hat{v}_{j,k} \Big) 
    \end{equation}    
    \begin{equation}
    z_{k} = \Big(\sum_{i=1}^n r_{i,k} + \sum_{j=1}^{\noisen} s_{j,k}\Big)     
    \end{equation}
    
    where $(y_k, z_{k})$ is the output for prover $\Pv_k$
    }

    \item[Line 12:]{Using the common public randomness $\{ b_{j,k} \}_{j \in [\noisen]}$ generated by $\oracle_{\texttt{morra}}$, the verifier updates their view of received commitments as follows:

      \[   
\hat{c}'_{j, k} = 
     \begin{cases}
       \Com(1, 0) \times {c'}_{j, k}^{-1} & \text{ if } b_{j, k}=1\\
       c'_{j, k}&\quad\text{otherwise.} \\ 
     \end{cases}
\]
    Note that $\Pv_k$ never opens ${c'}_{j, k}$, and thus $\Vfr$ never sees $\hat{v'}_{j, k}$ in plaintext. By the hiding property of commitments, an efficient verifier learns nothing about the prover's private values from these messages. However, as the update conditioned on $b_{j, k}$ is linear and $b_{j, k}$ is public, $\Vfr$ can still compute a commitment to $1 - v_{j, k}$ without ever knowing $v_{j, k}$. As a direct consequence, as discussed in the soundness claim, the prover cannot deviate from its prescribed linear operation, as the verifier is able to check it.
    As we will show later, this step guarantees correctness, soundness and security. 
    }
    \item[Line 13:]{Finally, the  verifier checks 
    \begin{equation}
    \label{eq:vfr_check}
     \prod_{i=1}^{n} c_{i,k} \times \prod_{j=1}^{\noisen} \hat{c}'_{j, k} = \Com(y_k, z_k)   
    \end{equation}        
    }    
\end{enumerate}

From these outputs, we can derive the desired result: we treat the $y_k$'s as shares, and calculate $y = \sum_{k=1}^{K} y_k$ as the noisy sum.  
We next show that this protocol achieves our desired properties.

%% file: tex_files/completeness.tex
\item{\textbf{Completeness:} 
By the definition of $\oracle_{\texttt{morra}}$, $(b_{1,k}, \dots, b_{\noisen,k})$ are all unbiased bits. 
As per $\Pi_{\texttt{Bin}}$, when $b_{j,k}=1$, $\hat{v}_{j,k}=1-v_{j,k}$ and when $b_{j,k}=0$, $\hat{v}_{j,k} = v_{j,k}$. 
We know that an honest prover is guaranteed to have sampled a private value $v_{j,k} \in \bit$ for all $j \in [\noisen]$. 
Thus the case-wise arithmetic operation described above is equivalent to setting $\hat{v}_{j,k} = v_{j,k} \oplus b_{j,k}$. 
This implies that for each server $\hat{v}_{j,k} \xleftarrow{R} \bit$ and $\sum_{j=1}^{\noisen} \hat{v}_{j,k} \sim \texttt{Binomial}(\noisen, 1/2)$. 
The output of each honest prover is thus $y_k = \texttt{Binomial}(\noisen, 1/2) + \sum_{i=1}^n \share{x_i}_k$. 
By linearity of secret-sharing, $\sum_{k \in [K]} y_k = \mathcal{M}_{\bin}(X, Q)$ where  $\mathcal{M}_{\bin}$ is defined in equation \eqref{eq:M_bin}. 
} 

%% file: tex_files/audit.tex
\item{\textbf{Soundness:}
Beyond exiting the protocol early (which is trivially detected), an adversary $\AdvA$ controlling a collection of dishonest provers could force a prover to cheat by doing at least one of the following:

\begin{enumerate}

    \item{ (Cheat at Line 4): For any $j \in [\noisen]$, $c'_{j,k}$ is not a commitment to a bit. As the verifier has access to oracle $\oracle_{\texttt{OR}}$, it would detect this immediately. Thus we can be guaranteed that $c'_{j,k}$ are commitments to 1 or 0.}

    \item{(Cheat at Line 7): The prover could sample improper public randomness. 
    However, this is impossible as the verifier and prover jointly use $\oracle_{\texttt{Morra}}$ to generate randomness.}
    
    \item{(Cheat at Line 10): Output messages $(y_k^\prime \neq y_k$, $z_{k}^\prime \neq z_{k})$. 
    If the verifier check from (Line 12) fails then the verifier knows $\Pv_k^*$ cheated. If 
      $\Com(y_k, z_{k}) = \prod_{i=1}^{n} c_{i, k}\times \prod_{j=1}^{\noisen} \hat{c}'_{j,k} = \Com(y_k^\prime, z_{k}^\prime)$, then $\AdvA$ has broken the binding property of the commitment scheme. As $\AdvA$ has negligible success in winning the discrete log game, it has a negligible chance at breaking the commitment scheme. 
    }
\end{enumerate}

These are the only places where the $\Pv^*$ sends a message to the $\Vfr$ and thus we have our result.
}

%% file: tex_files/zero_knowledge.tex
\item{\textbf{Zero-Knowledge:}
To prove zero knowledge we need to explicitly define the commitment scheme we are using. We use Pedersen Commitments which are defined as follows 
\begin{equation}
\label{eq:ped_com}
    \Com(x, r) = g^xh^{r}
\end{equation}

where $\mathcal{R}_\pp = \mathcal{M}_\pp= \Z_q$ and $\mathcal{C}_\pp = \G_q$ an abelian group where the discrete log problem is hard. To enhance readability, we will prove security for  $K=2$ provers and one verifier, but the result trivially generalises to $K \geq 2$ provers. To avoid confusion between the MPC and single curator setting, we defer the simpler security proof for single curators to Appendix \ref{app:single_curator_sec_proof}. 
Without loss of generality, assume that the verifier $\Vfr^*$ and $\Pv_1$ have been corrupted by a PPT adversary $\AdvA$ and that $\Pv_2$ is honest. 
$\Sim$ receives on its input tape the inputs for $\Pv_1$ and $\Vfr^*$. The ideal oracle functionality $\mathcal{M}_\bin$ is defined as before. Let $\Sim$ denote shorthand for $\Sim_{\Vfr^*, \Pv_1}$.  We construct the simulator as follows:

\begin{enumerate}
    
    \item{$\Sim$ receives the public messages $\Bigg\{  \Big\{ c_{i,k}  \Big\}_{k \in [K]} \Bigg\}_{i \in [n]}$ and sets $c_{i} = \prod_{k=1}^K c_{i,k}$.}

    \item{$\Sim$ internally invokes $\Pv_1$ to receive inputs $X_1$. If $\Pv_1$ was honest then $X_1 = \sum_{i=1}^n \llbracket x_i \rrbracket_1$. Of course, we have no control over $\AdvA$, and $X_1$ could be any arbitrary value. The definition of security requires that we prove security using the actual inputs used by the real-world adversary $\AdvA$ and not the ones it was handed to at the start of the protocol.}

    \item{$\Sim$ invokes $\mathcal{M}_{\texttt{bin}}$ with input $X_1$ and receives $(y, \Delta_1)$ as defined in equation \eqref{eq:M_bin}. Note $\Sim$ never has access to the honest party's input $X_2$ nor the randomness $\Delta_2$ used by $\Pv_2$ in the real protocol. It must simulate the messages and output of the real protocol from just its input and the output it receives from the ideal model.}

    \item{$\Sim$ sets $y_1= X_1 + \Delta_1$ and computes $y_2 = y - y_1$, which by the definition of $\mathcal{M}_\bin$, is equal to $(X_2 + \Delta_2)$.}

    \item{$\Sim$ samples $z_2 \xleftarrow[]{R} \mathcal{R}_\pp$ and sets $c_{2} = \Com(y_2, z_{2})$.}

    \item{$\Sim$ samples $c'_{2,2}, \dots, c'_{{\noisen},2}$ such that $c'_{j,2} = \Com(1, s_{j,2})$ where $s_{j,2} \xleftarrow[]{R} \mathcal{R}_\pp$. 
     It sets $c'_{1,2} = g^1a_2$ where $a_2 = c_{2} \times  \Big( \prod_{j=2}^{\noisen} \hat{c}'_{j,2}\Big)^{-1} \times \Big( \prod_{i=1}^{n} c_{i,2}\Big)^{-1} \times g^{-1}$. 
     Notice that $\Sim$ is actually unable to open $c'_{1,2}$ but is never required to do so, as opening a commitment to a private value violates DP. The only information $\AdvA$ can check is if $c'_{1,2}$ is a commitment to a bit, which it is. Thus the simulator artificially constructs a set of commitments that align like the real-world protocol, without having the slightest idea what the randomness used by $\Pv_2$ actually was. It is able to do so due to the hiding property of the commitment scheme.}
    
     \item{$\Sim$ sends over $\{c_{j,2}\}_{j \in [\noisen]}$ to $\AdvA$ pretending to be the honest prover (Line 4 of Figure \ref{fig:dp_interactive_proof}).}
	    
    \item{$\Sim$ pretends to be the prover and jointly invokes $\oracle_{\texttt{Morra}}$ with $\AdvA$ to sample $\noisen$ unbiased public bits $(b_{1,2}, \dots, b_{\noisen,2})$.}
    
    \item{Finally $\Sim$ sends $y_2$ and $z_{2}$ to $\AdvA$ and outputs whatever $\AdvA$ outputs.}\end{enumerate}}

%% file: tex_files/separation.tex
\section{Separation Under Verifiable DP}
\label{sec:separation}

We show that information theoretic verifiable DP is impossible in the trusted curator model.
To prove our result stated in Theorem~\ref{theorem:coms_imply_vdp}, we rely on the impossibility of secure coin flipping by \cite{haitner2014coin}.

\begin{theorem}[Impossibility Of Tossing A Fair Coin]\cite{haitner2014coin} 
 \label{theorem:coin_flip_implies_one_way_functions} Let $(\Pv, \Vfr)$ be a coin tossing protocol and let $B_\lambda = \mathbb{E}[\texttt{out}(\Pv, \Vfr)(1^\kappa)]$ be the bias of the output of such a protocol. Assuming that one-way-functions do not exist, then for any $g \in \texttt{poly}(\kappa)$, there exists a pair of efficient cheating strategies $\Pv^*$ and $\Vfr^*$ such that the following holds: for infinitely many $\kappa$'s, for each $j \in \bit$ either $\Pr[\texttt{out}(\Pv^*, \Vfr)(1^\kappa) = j]$ or $\Pr[\texttt{out}(\Pv, \Vfr^*)(1^\kappa)= j]$ is greater than $\sqrt{B_\kappa^j} - \frac{1}{g(\kappa)}$, where $B_\kappa^1 = B_\lambda$ and $B_\kappa^0 = 1 - B_\lambda$. In particular for $B_\lambda=\frac{1}{2}$, the corrupted party can bias the outcome by almost $\frac{1}{\sqrt{2}} - \frac{1}{2}$.
\end{theorem}

The theorem above states that it is impossible for two unbounded parties to jointly sample an unbiased public coin. The result is stronger than the impossibility result by Cleve \cite{cleve1986limits}, which states that its impossible to jointly flip an unbiased coin if we allow parties to exit early. 
The theorem above states that it's impossible even if we guarantee no party exists the protocol early.

\begin{theorem}[Information Theoretic Verifiable DP is impossible]  
\label{theorem:coms_imply_vdp}  
Any constant round interactive protocols $\Pi$ for an DP-mechanism $\mathcal{M}_\bin$ that satisfies Verifiable-DP (Definition \ref{defn:vdp_MPC}) cannot have unconditional soundness and statistical zero-knowledge.
\end{theorem}

\begin{proof}
%

Verifiable DP requires that a verifier be able to guarantee that the randomness generated by a prover remains unbiased, without the verifier ever seeing the randomness. Theorem \ref{theorem:coin_flip_implies_one_way_functions}, states that it is impossible for two unbounded parties to even jointly sample a \textit{public} unbiased coin without assuming one way functions. 
Thus commitment schemes are both necessary and sufficient to jointly sample an unbiased public coin. \par
The task of jointly sampling unbiased \textit{private} randomness is harder. If two parties could sample unbiased private randomness, then they could just use the same protocol to sample unbiased public randomness, by revealing the randomness.
Thus, commitment schemes are a necessary condition for verifiable DP.
Commitments cannot be both statistically binding and hiding, thus unbounded soundness and statistical zero-knowledge is impossible.
\end{proof}

\parag{Connection With Open Problem}
\begin{definition}[$\alpha$-useful mechanism]
  \label{defn:utility} Fix $\alpha \in [0,1]$. Let $u : \mathcal{X}^n \times \mathcal{Y} \rightarrow \in \bit$  be an efficiently computable deterministic function. A mechanism $\mathcal{M}$ is $\alpha$-useful for a utility function $u$ if for some $Q \in \mathcal{Q}$ and for all $X \in \mathcal{X}^n$
  \begin{equation}
  \Pr_{y \leftarrow \mathcal{M}(X, Q)}[u(X, y) = 1] \geq \alpha  
  \end{equation}
  \end{definition}

In his survey on the complexity of DP, Vadhan \cite{vadhan2017complexity} asks the following question. Given $X \in \mathcal{X}^n$ and a differentially private mechanism $\mathcal{M}: \mathcal{X}^n \times \mathcal{Q} \rightarrow \mathcal{Y}$, is there an efficient utility function $u$ that is $\alpha$-useful when $\mathcal{M}$ is IND-CDP but not when $\mathcal{M}$ is information-theoretically DP. 
Groce \etal \cite{groce2011limits} show that if the range of $u$ is in $\mathcal{R}^n$ and the utility is measured in terms of the $\mathcal{L}_p$-norm, then statistical-DP and computational DP are equivalent. Thus for the separation to hold, the range of $u$ must have a more complex structure, such as a graph, a circuit or a proof. 
Bun \etal corroborate this result by describing a utility function such that $u$ is infeasible (not impossible) when $\mathcal{M}$ is statistical DP and efficient when $\mathcal{M}$ is computational DP~\cite{bun2016separating}. 
Similar to our definition of verifiability, their utility function $u$ is cryptographic and unnatural from a data analysis point of view. Specifically, given $y = \mathcal{M}(X, Q)$, Bun \etal, define the utility as the answer to the question of whether $y$ is a valid zap proof \cite{dwork2000zaps} of the statement ``there exists a row in $X$ that is a valid message signature pair''. 
 Meanwhile, we define our utility function as an interactive proof, that checks whether the real protocol output $y$, is indistinguishable from the output of an ideal run of $\mathcal{M}$. 
 In Theorem \ref{theorem:coms_imply_vdp}, we show that verifiable DP is impossible in the presence of computationally unbounded adversaries. 
This provides a candidate for a separation between statistical DP and computational DP.
 
 However, there are some key differences between our formulation of utility and how it was originally posed.
 For example, in Bun \etal, the utility function $u$ is a deterministic non-interactive function that receives the output $y$ and a dataset $X$ of message-signature pairs. The task of evaluating utility is separate from the task of computing DP statistics.
 In verifiable DP, both the DP statistic and utility are computed simultaneously via a constant round interactive protocol. 
 Furthermore, the number of rounds of the utility function is a function of the privacy parameter $\epsilon$. 
 Another point of difference is that, in verifiable DP, the verifier performs the dual role of evaluating the utility of the mechanism and generating randomness that prevents a curator from cheating (although it does not ever see this randomness). In Bun \etal, the verifier's task is just to verify the proof. 
 They are not involved in generating the DP noise.
 Although we show that information theoretic verifiable DP is impossible, our definitions allow the adversary more agency. Thus the two settings are not directly comparable. 
 We defer finding stronger connections between verifiable DP and finding a utility function that separates DP as per  \cite{vadhan2017complexity} to future work.

%% file: tex_files/performance.tex
\section{Performance}
\label{sec:expts}

This section quantifies the computational cost of $\Pi_\bin$, our protocol for computing verifiable DP counting queries. 
In the trusted curator model, the non-verifiable protocol simply involves summing over $n$ inputs, sampling one draw of Binomial noise and aggregating the results.
 Meanwhile, verifiable DP requires computing commitments for $\noisen$ private coins, sending the verifier a non-interactive OR proof (as described in Appendix \ref{app:sigma_open}) that the messages are commitments to either zero or one, playing $\noisen$ rounds of $\Pi_{\morra}$ and performing exponentiation operations to check if the prover messages align.  
In our experiments, we instantiate the commitment scheme using Pedersen Commitments (PC) \cite{pedersen1991non} 
and consider two choices for the prime order group $\G_q$ for PC. 
In what follows, we adopted  $\G_q \subset \Z_p^*$ based on the finite field discrete log problem\footnote{Implemented using \url{https://docs.rs/openssl/latest/openssl/bn/struct.BigNum.html}}. 
We also implemented Pedersen commitments over elliptic curves using the prime order Ristretto 
group\footnote{\url{https://doc.dalek.rs/curve25519_dalek/ristretto/struct.RistrettoPoint.html}}, which gave slower results. 
A single exponentiation operation on an 8 core Apple M1 Mac took 35 $\mu s$ for $\G_q \subset \Z_p^*$ and 328 $\mu s$ over Curve25519.  
 All results and plots are reproducible using code found at \url{https://anonymous.4open.science/r/Verifiable-Differential-Privacy-0407/README.md}.

 \begin{table}[]
 \caption{The table above benchmarks the latency of each of the different stages of $\Pi_\bin$ for computing single dimension counting queries with parameters $n=10^6, \epsilon=1.25, \delta= 2^{-10} $.  Setting $\epsilon=0.88$ results in $\noisen=262144$ private coins for DP.  }
\label{tab:pipeline}
\begin{tabular}{|l|l|l|l|l|}
\hline
$\Sigma$-proof & $\Sigma$-verification & Morra& Aggregation  & Check     \\ \hline
6609 ms        & 6708 ms               & 4987 ms                        & 198 ms                & 263 ms \\ \hline
\end{tabular}
\end{table}

 \begin{figure}[t]
     \centering    
     \includegraphics[scale=0.9]{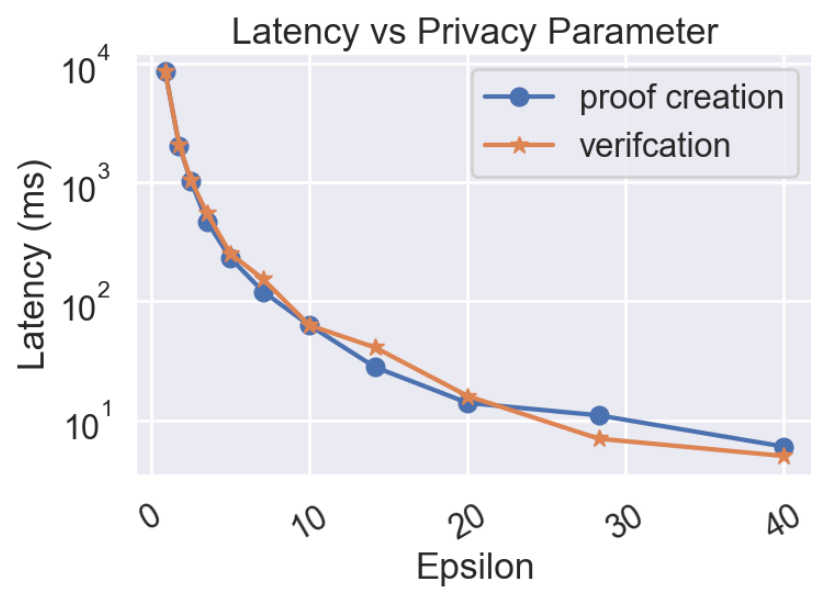}
     \caption{The figure above describes the latency of the two most expensive operations in $\Pi_\bin$ - the time it takes to prove and validate that the prover's private value is indeed a bit. The smaller the privacy parameter $\epsilon$ (high privacy), the more private coins $\noisen$ are needed to provide DP. }
     \label{fig:esp_vs_coins}
 \end{figure}

\parag{Time cost for server verification}
In Table \ref{tab:pipeline} we describe the latency of the different phases of $\Pi_\bin$. The protocol's main computational bottleneck is verifying that the commitments to the prover's private values are indeed commitments to bits. 
Thus most of the time is spent creating and verifying non-interactive $\Sigma$-proofs. 
The aggregation column describes how long the prover takes to aggregate $n$ elements of size $\kappa$ bits, and the check column describes the time it takes the verifier to compute commitments to check that the prover's messages align with the commitments.
 The Morra column describes the time it takes to sample $\noisen$ public coins using $\Pi_\morra$.
 We measure time spent doing local computations, and do not include time spent to communicate over the network.
As working with the $\Sigma$-proof is our main bottleneck, 
Figure~\ref{fig:esp_vs_coins} describes how proof creation and verification latency scales with the privacy parameter $\epsilon$.  Note that for high privacy settings (small values of $\epsilon$), the prover(s) need to generate more private coins to ensure indistinguishability.
Specifically, the number of coins ($n_b$) is proportional to $1/\epsilon^2$ (Lemma~\ref{theorem:dp_guarantee}), and the time cost is then linear in $n_b$. 

\parag{Time cost for client verification}
 Clients submit secret shares of their inputs in the MPC setting. Thus the servers must verify that the client inputs are valid. For $M$-dimensional DP-histogram estimation, the client inputs are restricted to one-hot encoded vectors of size $M$. As discussed in Section \ref{sec:client_verification}, the sketching techniques used in PRIO and Poplar allow servers to verify clients with information-theoretic security but are vulnerable to attack by malicious servers.
  Our use of $\Sigma$-OR-protocols can defend against such attacks, but it comes at a higher computational cost due to its reliance on public key cryptography. 
 Figure~\ref{fig:client_verification} benchmarks the increase in latency as a function of the number of input dimensions ($M$) of the secret shares. 
 In both cases, the cost grows with the dimensionality of the input.  The cost of making client verification robust to malicious servers is approximately an order of magnitude. 
 The $\Sigma$ protocol for verification can be run on each input dimension in parallel, and thus computation can be sped up using more cores.
 However, this increases the communication bandwidth of the protocol.

 \begin{figure}[t]
     \centering    
     \includegraphics[scale=0.9]{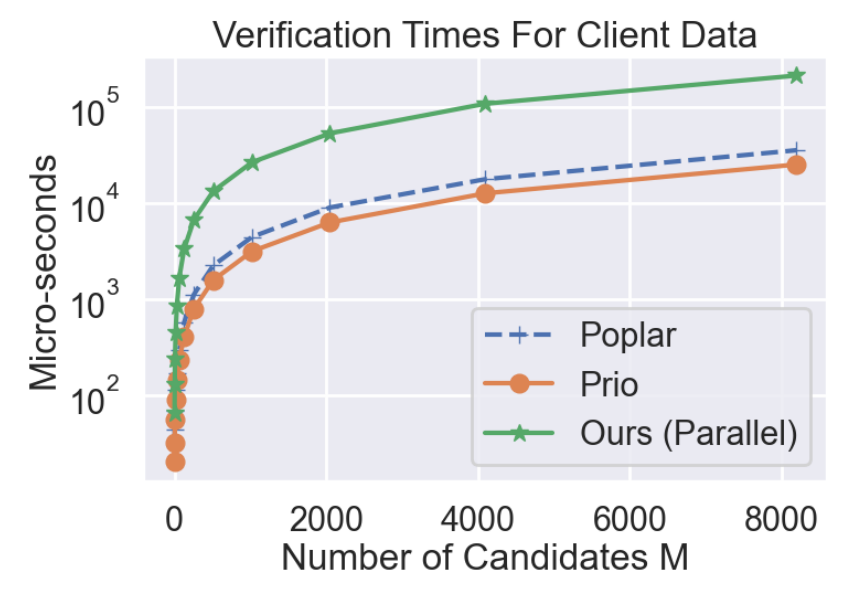}
     \caption{The figure above compares the times it takes to verify if a single client input is valid.  PRIO and Poplar use lightweight sketching protocols and general-purpose MPC to check in zero knowledge whether a client's input is a one-hot vector. Instead, we use the $\Sigma$-OR proof to check if each coordinate of a client is legal.}
     \label{fig:client_verification}
 \end{figure}

%% file: tex_files/related_work.tex
\pdfoutput=1

\section{Related Work}
\label{sec:related_work}
\UseRawInputEncoding


\begin{table*}[]
\centering
\caption{Summary of efforts MPC computation of aggregate DP statistics. 
The active security column describes if the protocols allowed participants to deviate arbitrarily. The Central DP column describes if the protocol output satisfies constant DP error independent of the number of clients participating in the protocol. The auditable property describes if the final output can be verified for correctness. Some interactive protocols leak additional information (such as prefix information about client input bits) beyond just the DP output. The leakage column describes if the prescribed protocols suffered from additional leakage.}
\label{table:related_work}
\begin{tabular}{@{}lllll@{}}
\toprule
Protocol                                   & Active Security                                            & Central DP                        & Auditable                         & Zero Leakage                           \\ \midrule
\multicolumn{1}{|l|}{Cryptographic RR \cite{ambainis2004cryptographic}}  & \multicolumn{1}{c|}{$\checkmark$}                &  \multicolumn{1}{l|}{} &   \multicolumn{1}{c|}{}   & \multicolumn{1}{c|}{$\checkmark$}     \\
\multicolumn{1}{|l|}{Verifiable Randomization Mechanism \cite{kato2021preventing}}  & \multicolumn{1}{c|}{$\checkmark$}                &  \multicolumn{1}{l|}{} &   \multicolumn{1}{c|}{$\checkmark$}   & \multicolumn{1}{c|}{$\checkmark$}     \\
\multicolumn{1}{|l|}{Securely Sampling Biased Coins \cite{champion2019securely}}              & \multicolumn{1}{l|}{}                & \multicolumn{1}{c|}{$\checkmark$} & \multicolumn{1}{l|}{$\xmark$}     & \multicolumn{1}{c|}{$\checkmark$}     \\
\multicolumn{1}{|l|}{MPC-DP heavy hitters\cite{bohler2021secure}}              & \multicolumn{1}{l|}{}                & \multicolumn{1}{c|}{$\checkmark$} & \multicolumn{1}{l|}{$\xmark$}     & \multicolumn{1}{c|}{$\checkmark$}     \\
\multicolumn{1}{|l|}{PRIO \cite{corrigan-gibbs_prio_2017}}              & \multicolumn{1}{l|}{}                & \multicolumn{1}{c|}{$\checkmark$} & \multicolumn{1}{l|}{$\xmark$}     & \multicolumn{1}{c|}{$\checkmark$}     \\
\multicolumn{1}{|l|}{Brave STAR \cite{davidson2021star}}        & \multicolumn{1}{l|}{}                 & \multicolumn{1}{l|}{}             & \multicolumn{1}{l|}{$\xmark$}     & \multicolumn{1}{l|}{} \\
\multicolumn{1}{|l|}{Sparse Histograms \cite{bell2020secure}} & \multicolumn{1}{l|}{}                 & \multicolumn{1}{c|}{$\checkmark$} & \multicolumn{1}{l|}{$\xmark$}     & \multicolumn{1}{l|}{} \\
\multicolumn{1}{|l|}{Crypt-$\epsilon$ \cite{roy2020crypt}}           & \multicolumn{1}{l|}{}                 & \multicolumn{1}{c|}{$\checkmark$} & \multicolumn{1}{l|}{$\xmark$}     & \multicolumn{1}{l|}{}             \\
\multicolumn{1}{|l|}{Poplar \cite{boneh_lightweight_2022}}            & \multicolumn{1}{c|}{$\checkmark$} & \multicolumn{1}{c|}{$\checkmark$} & \multicolumn{1}{l|}{$\xmark$}     & \multicolumn{1}{l|}{} \\
\multicolumn{1}{|l|}{Our work}          & \multicolumn{1}{c|}{$\checkmark$}                     & \multicolumn{1}{c|}{$\checkmark$} & \multicolumn{1}{c|}{$\checkmark$} & \multicolumn{1}{c|}{$\checkmark$}     \\ \bottomrule
\end{tabular}
\end{table*}

Dwork \etal introduced DP and described the Laplace mechanism for outputting histograms in the trusted curator model~\cite{dwork2006calibrating}.  Soon after, McSherry \etal proposed the exponential mechanism \cite{mcsherry2007mechanism} (equivalently, report noisy max \cite{ding2021permute}), which lets us compute the (approximately) most frequent bucket in a histogram, also under pure differential privacy.  Although these mechanisms give us pure differential privacy and optimal error rates $O(\frac{1}{\epsilon})$, implementing such a ``central'' model requires trusting that the curator to follow the protocol and not exploit the client data that it sees in plaintext.

Therefore, researchers studied local privacy (LDP)  \cite{kasiviswanathan2011can} using randomised response \cite{warner1965randomized} to prevent any other party from seeing data in plaintext. 
Recently, Cheu, Smith and Ullman showed that the randomised response algorithm generalises all locally private protocols \cite{cheu2021manipulation}. 
This generalisation highlights two unavoidable disadvantages of local differential privacy. 
The first is that the accuracy of the protocol for even the binary histogram is $O(\sqrt{n})$ compared to $O(1)$ in the central model. 
The second is that randomised response systems offer a much weaker definition of privacy than the usual cryptography standards such as semantic security. 
For example, if the client flips their original answer with probability $p=0.1$, the curator sees their sensitive information in plain text 90\% of the time. 
Further increasing $p$ reduces the accuracy of the protocol dramatically. 
Consider the example from~\cite{corrigan-gibbs_prio_2017}, where 1\% of a million people answer ``yes'' to a survey about a sensitive topic. 
If we set $p=0.49$, then one-third of the time the central analyser concludes that not a single  member of the population answered ``yes''. 
Thus if we want to preserve utility, this definition of security is considerably weaker than the indistinguishability guarantees provided by protocols such as secret sharing.

Shuffle privacy \cite{cheu2021differential, balle2019privacy, erlingsson_amplification_2020} analyses local mechanisms under the lens of central privacy and bridges the accuracy gap between local and central models. Recent results~\cite{ghazi2019power, balcer_separating_2020} prove that near central error guarantees are possible with distributed local transformations. Although this bypasses the accuracy issue of LDP, shuffle privacy assumes the existence of a secure shuffler, which is non-trivial to implement. In recent work, Bell \etal show that secure aggregation realises secure shuffling \cite{bell2020secure}. 
However, such protocols impose the impractical constraint of secure peer-to-peer communication between clients, and the curator is still a single source of failure. 
Despite the immense progress on differentially private histogram estimation, all known efficient implementations assume semi-honest participants and are a variant of either randomised response or the additive mechanism. 
It only takes a small fraction of clients to deviate from their prescribed protocol to destroy any utility of randomised response \cite{cheu2021manipulation}. 
Additive mechanisms involve adding carefully curated randomness to the statistic before being released as output. 


To ensure central DP error without a trusted curator, Dwork \etal proposed using standard MPC for computing DP statistics~\cite{dwork2006our}. 
They proposed that each of the $K$ servers would own a fraction of the entire dataset used for computation. 
As long as not more than $\lfloor \frac{K}{3} \rfloor$ of the servers are dishonest, it is possible to compute DP-histograms with optimal accuracy. 
However, the protocol is not publicly auditable and breaks down in presence of a dishonest majority of adversarial corruptions. McGregor \etal show a separation between DP obtained using a trusted curator and that obtained using MPC \cite{mcgregor2010limits}. Specifically, they show that there exist computations (such as inner product or hamming distance) where mechanisms with $(1, 0)$-DP incur $\Omega(\sqrt{n})$ reconstruction error compared to $O(1)$ in presence of a trusted curator. 
To bridge this gap, Mironov \etal defined computational differential privacy, a relaxation of traditional DP~\cite{mironov2009computational}. They show that as long as semi-honest OT exists, it is possible to compute any computationally DP function with the same error rates as information theoretic DP in a trusted curator model. 
Histograms, unlike inner product and hamming distance, can be computed using MPC with the same error rates as trusted curator DP, under infomation theoretic DP. 
Thus recent work has focused on computing histograms using MPC. 

Bohler \etal use MPC to compute heavy hitters with semi-honest adversaries~\cite{bohler2021secure}. 
Researchers at Brave use oblivious pseudorandom functions (OPRF's) \cite{jarecki2009efficient} and Shamir secret sharing \cite{shamir1979share} to compute $k$-anonymous histograms in the two server setting~\cite{davidson2021star}.  However, they do not include support for differential privacy. Researchers at Google use linear homomorphic encryption and OPRFs to compute differentially private sparse histograms in two-server models (2PC)~\cite{bell2022distributed}, but require both the servers and clients to be semi-honest.
Gibbs and Boneh propose PRIO, a protocol in which a small number of servers receive arithmetic shares of client input to compute differentially private histograms \cite{corrigan-gibbs_prio_2017}. PRIO uses shared non interactive proofs (SNIP's) to prevent clients from submitting illegal inputs but the protocol is only honest-verifier zero knowledge. Following the popularity of PRIO, Addanki \etal introduce PRIO+ to work over Boolean shares \cite{addanki2022prio+}. Boneh \etal use distributed point functions (DPFs) \cite{boyle2019secure} to compute DP heavy-hitters in the two server model to propose a system called Poplar \cite{boneh_lightweight_2022} that is zero knowledge even in presence of active adversaries. Roy \etal introduce \textit{Crypt}-$\epsilon$, a generic system to compute differentially private statisitcs using garbled circuits and linear homomorphic encryption \cite{roy2020crypt}. The general purpose natue of \textit{Crypt}-$\epsilon$ guarantees security only in the semi-honest threat model. Ambainis \etal proposed cryptographic randomised response \cite{ambainis2004cryptographic} but are able to only guarantee local differential privacy. Table \ref{table:related_work} summarises the assumptions under which the latest MPC protocols that have been used to compute DP statistics. As described earlier, existing work either assumes semi-honest adversaries or is not auditable. In 2021, the State Of Alabama sued the US deparment of commerce with regard to the errors caused due to random noise \cite{courtCase}. Differential Privacy by its defintion introduces a random noise blanket that tradesoff accuracy for privacy. This randomness is unavoidable if we wanted to protect individual privacy, but it also enables a corrupt aggregating server to disguise adversarial behaviour as randomness. In our paper, we first upgrade to security against active adversaries. Like existing literature we work in the dishonest majority model and further require the protocols to be publicly auditable. Our privacy constraints describe the most strict adversarial setting for practical deployment.

\section{Concluding Remarks}

We have introduced the notion of verifiable differential privacy to prevent malicious aggregators using random noise as an attack vector. 
We have demonstrated feasibility of this notion, and showed that computational DP is necessary to achieve verifiability. 
A natural open question is to provide protocols for more complex DP mechanisms.  
Our protocol deliberately uses simple randomness (a Binomial distribution constructed from Bernoulli random variables), as making verifiable Laplace or Gaussian noise is far from clear.  
Similarly, approaches based on sampling from an appropriate distribution (the exponential mechanism) may be challenging, since the distribution itself leaks information about the private data. 
Another direction would be to support our approach within more expressive MPC frameworks (e.g., auditable-SPDZ~\cite{baum2014publicly}).

%% file: Appendix/a_security_proof.tex
\section{Formal Security Definitions}
\label{app:sec_defns}

\begin{definition} [Discrete Log Assumption]
For all PPT adversaries $\AdvA$, there exists a negligible function $\mu$ such that 
     \[ \Pr\left[ x = x^\prime : \begin{array}{c} (\G_q, g) \leftarrow \texttt{Setup}(1^\kappa) \\
     x \xleftarrow[]{R} \Z_q, h = g^x \\
     x^\prime \leftarrow \AdvA(\pp, h) \\
    \end{array} \right]
  \leq \mu(\kappa)  \] 

\label{defn:discrete_log}
\end{definition}

\begin{definition} (Hiding Commitments) Let $\kappa$ be the security parameter. A commitment scheme is said to be hiding for all PPT adversaries $\AdvA$ the following quantity is negligible. The commitment is perfectly hiding if $\mu(\kappa) = 0$.
\label{defn:hiding_com}
     \[ \Pr\left[ b = b^\prime : \begin{array}{c} \pp \leftarrow \texttt{Setup}(1^\kappa) \\
     b \xleftarrow[]{R} \bit, r_{x_b} \xleftarrow[]{R} \texttt{R}_\pp \\
     (x_0, x_1) \in \mathcal{M}_\pp^2 \leftarrow \AdvA(\pp) \\
     c = \Com(x_b, r_{x_b}), b^\prime = \AdvA(\pp, c)
    \end{array} \right]
  \leq \mu(\kappa)  \] 
\end{definition}

\begin{definition} (Binding Commitments) Let $\kappa$ be the security parameter. A commitment scheme is said to be binding if, for all PPT adversaries $\AdvA$, there exists a negligible function $\mu$ such that
\label{defn:binding_com}
     \[ \Pr\left[ (c_{x_0}= c_{x_1}) \land (x_0 \neq x_1): \begin{array}{c} \pp \leftarrow \texttt{Setup}(1^\kappa) \\
     x_0, r_{x_0}, x_1, r_{x_1} \leftarrow \AdvA(\pp)
    \end{array} \right]
  \leq \mu(\kappa)  \] 
The commitment is perfectly binding if $\mu(\kappa) = 0$.
\end{definition}

%% file: Appendix/b_dp_proofs.tex
\section{The Binomial mechanism}
\label{app:dp_proof_bin_mech}

In this Appendix, we spell out the details of the differential privacy properties of Binomial noise addition (the Binomial mechanism). 
The results here were originally shown by Ghazi \etal~\cite{ghazi_power_2020}, and we include them here for completeness (we do not claim any novelty in this section). 

\begin{definition}
\label{def: k-incremental} \cite{ghazi_power_2020} A function $q: \mathcal{X}^n \rightarrow \Z^M$ is said to be $k$-incremental if for all neighbouring datasets $X \sim X^{\prime}$, $|| f(X) - f(X^{\prime})||_{\infty} \leq k$.
\end{definition}

It is easy to see that counting queries for histogram estimation are 1-incremental. The following definition describes valid noise distributions to ensure differential privacy.

\begin{definition}
\label{def: smooth} \cite{ghazi_power_2020} A distribution $D$ over $\Z$ is $(\epsilon, \delta, k)$-smooth if for all $k^{\prime} \in [-k, k]$ we have 
\begin{equation}
    \Pr_{Y \sim D} \Bigg[ \frac{\Pr_{Y^{\prime} \sim D}[Y^{\prime} = Y]}{\Pr_{Y^{\prime} \sim D}[Y^{\prime} = Y+k^{\prime}]} \geq e^{|k^{\prime}|\epsilon}\Bigg] \leq \delta
\end{equation}
\end{definition}

The result for the Binomial mechanism follows by showing that adding noise drawn from a smooth distribution ensures differential privacy and then showing that the Binomial distribution meets the smoothness definition.

\begin{lemma}[Lemma 4.11 in Appendix~C of \cite{ghazi_power_2020}]
\label{lemma:smooth_distributions}
Suppose $q: \mathcal{X}^n \rightarrow \Z^M$ is $k$-incremental i.e., for all neighboring datasets $X \sim X^{\prime}$ we have $||q(X) - q(X^{\prime}) ||_{\infty} = k$ and $\Delta(q) = || q(X) - q(X^{\prime}) ||_1 = \Delta$. Let $\mathcal{D}$ be a $(\epsilon, \delta, k)$-smooth distribution. Then the mechanism $M$
\begin{equation}
    M_{(Y_1, \dots, Y_M)}(X, q) = q(X) + (Y_1, \dots, Y_M)
\end{equation}
\noindent
is $(\epsilon\Delta, \delta\Delta)$ differentially private, where $(Y_1, \dots, Y_M) \overset{\text{i.i.d}}{\sim} D.$
\end{lemma}

\begin{proof}
Let $X = (x_1, \dots, x_n)$ where $x_i \in \mathcal{X}$ and $X^\prime = (x_1, \dots, x_n^\prime)$. 
Let $\vec{y} = (y_1, \dots, y_M)$ and $\vec{Y} = (Y_1, \dots, Y_M)$ be i.i.d draws from $\mathcal{D}$. Assume that Equation \eqref{eq: smooth_dp} holds:
\begin{equation}
\label{eq: smooth_dp}
\Pr_{\vec{y} \sim \mathcal{D}}{[g(\vec{y}) \geq e^{\epsilon^\prime}]} \leq \delta^\prime
\end{equation}
\noindent
where $g(\vec{y}) = \frac{\Pr_{(\vec{Y} \sim \mathcal{D})}{[M_{\vec{Y}}(X, q) = q(X) + \vec{y} ]}}{\Pr_{(\vec{Y} \sim \mathcal{D})}{[M_{\vec{Y}}(X^\prime, q) = q(X^\prime) + \vec{y} ]}}$.

Let $S \subseteq \Z^M$ be an arbitrary subset in the range of $M$. Let $T = \{M_{\vec{Y}}(X, q) | g(\vec{y}) < e^{\epsilon^\prime}\}$ represent a set of outputs of $M$ over draws of $\vec{Y}$ such that $g(\vec{y}) < e^{\epsilon^\prime}$. 
Then from equation \eqref{eq: smooth_dp} we can show that $M$ is $(\epsilon^\prime, \delta^\prime)$ differentially private
\begin{align}
\Pr_{\vec{y} \sim \mathcal{D}}[M_{\vec{y}}(X,& q) \in S] \leq \delta^\prime + \Pr_{\vec{y} \sim \mathcal{D}}{[M_{\vec{y}}(X, q) \in S \cap T]}  \label{eq:smooth_dp_a}\\
&= \delta^\prime + \sum_{w \in S \cap T }\Pr_{\vec{y} \sim \mathcal{D}}{[M_{\vec{y}}(X, q) = w]} \\
&\leq \delta^\prime + \sum_{w \in S \cap T }e^{\epsilon^\prime}\Pr_{\vec{y} \sim \mathcal{D}}{[M_{\vec{y}}(X^\prime, q) = w]} \label{eq:smooth_dp_b}\\
&\leq \delta^\prime + \sum_{w \in S}e^{\epsilon^\prime}\Pr_{\vec{y} \sim \mathcal{D}}{[M_{\vec{y}}(X^\prime, q) = w]} \label{eq:smooth_dp_c}\\
&= \delta^\prime + e^{\epsilon^\prime}\Pr_{\vec{y} \sim \mathcal{D}}{[M_{\vec{y}}(X^\prime, q) \in S]}
\end{align}

Equation $\eqref{eq:smooth_dp_a}$ is from the law of total probability, equation $\eqref{eq:smooth_dp_b}$ comes from equation \eqref{eq: smooth_dp} assumption and equation \eqref{eq:smooth_dp_c} is true as $T \cap S \subseteq S$. 
Therefore all that remains is to show that equation \eqref{eq: smooth_dp} is true if $\mathcal{D}$ is as defined and $\epsilon^\prime = \epsilon\Delta, \delta^\prime= \delta\Delta$ to complete the proof.

Define $k_j = q(X)_j - q(X^\prime)_j$. As each coordinate of $q(X)$ is independently perturbed and $q$ is a deterministic function, equation \eqref{eq: smooth_dp} is equivalent to equation \eqref{eq: smooth_dp_eqiv}.
\begin{equation}
\label{eq: smooth_dp_eqiv}
\Pr_{\vec{y} \sim \mathcal{D}}{\Bigg[  \prod_{j=1}^M
\frac{\Pr_{(Y_j \sim \mathcal{D})}{[Y_j = y_j ]}}{\Pr_{(Y_j \sim \mathcal{D})}{[Y_j = y_j + k_j ]}} \geq e^{\epsilon^\prime}\Bigg]} \leq \delta^\prime
\end{equation}
Thus, in order to prove equation \eqref{eq: smooth_dp} is true, it suffices to show that equation \eqref{eq: smooth_dp_eqiv} holds. 
We know that $\mathcal{D}$ is a smooth distribution i.e for each $j \in [M]$
\begin{align}
\Pr_{\vec{y} \sim \mathcal{D}}{\Bigg[ 
\frac{\Pr_{(Y_j \sim \mathcal{D})}{[Y_j = y_j ]}}{\Pr_{(Y_j \sim \mathcal{D})}{[Y_j = y_j + k_j ]}} \geq e^{|k_j|\epsilon}\Bigg]} \leq \delta
\end{align}
We can apply the union bound to get the probability of the joint distribution over all indices.

\begin{align}
\Pr_{\vec{y} \sim \mathcal{D}}{\Bigg[ \prod_{j=1}^M
\frac{\Pr_{(Y_j \sim \mathcal{D})}{[Y_j = y_j ]}}{\Pr_{(Y_j \sim \mathcal{D})}{[Y_j = y_j + k_j ]}} \geq e^{\sum_{j=1}^m|k_j|\epsilon}\Bigg]}  \\
\leq \delta \sum_{j=1}^M \mathbb{I}(k_j \neq 0)
\end{align}

Given the sensitivity of $q$ is $\Delta$, at most $\Delta$ indices for $k_j$ can be non-zero and $(\sum_{j=1}^M |k_j|) \leq \Delta$. Finally we get the result we seek
\begin{align}
\Pr_{\vec{y} \sim \mathcal{D}}{\Bigg[ \prod_{j=1}^M
\frac{\Pr_{(Y_j \sim \mathcal{D})}{ [Y_j = y_j ]}}{\Pr_{(Y_j \sim \mathcal{D})}{[Y_j = y_j + k_j ]}} \geq e^{\Delta|k_j|\epsilon}\Bigg]} \leq \delta \Delta \label{eq: non_zero}
\end{align}
\end{proof}

\begin{lemma} [Based on Lemma 4.12 of Appendix~C in \cite{ghazi_power_2020}]
\label{lemma: bin_is_k_smooth}
Let $n \in \N$, $p \in [0, 1/2]$, $\alpha \in [0,1)$ and $k \leq \frac{n\alpha p}{2}$. Then the binomial distribution $\Bin(n, p)$ is a $(\epsilon, \delta, k)$-smooth distribution.
\end{lemma}

\begin{proof}
Let $Y \sim \Bin(n, p)$, then $\Pr{[Y = y]} = \binom{n}{y}p^{y}(1-p)^{n-y}$. For any $-k \leq k^\prime \leq k$, define an interval $\mathcal{\epsilon} := [(1 - \alpha)np + k^\prime, (1 + \alpha)np - k^\prime]$. This an interval of size $k$ around the mean of the distribution. Note that as long as $k \leq \frac{np}{2}\alpha$, then the interval $\mathcal{\epsilon^\prime} := [(1 - \alpha/2)np, (1 + \alpha/2)np]$ is contained inside of $\mathcal{\epsilon}$. 
Thus if $y \sim \Bin(np)$ is not in $\mathcal{\epsilon}$, it is also not inside $\mathcal{\epsilon^\prime}$. We know how to bound the probability that $y \notin \mathcal{\epsilon^\prime}$ by using the multiplicative Chernoff bound. 
Invoking it, we get
\begin{align*}
\Pr_{y \sim \Bin(n,p)}{[y \notin \epsilon]} &\leq \Pr_{y \sim \Bin(n,p)}{[y \notin \epsilon^\prime]} \\
&\leq e^{-\frac{-\alpha^2np}{8}} + e^{-\frac{-\alpha^2np}{8+2\alpha}} \\
&= \delta \\
\end{align*}

Now for all $y \in \mathcal{\epsilon}$, we have for $0 \leq k^\prime \leq k$
\begin{align}
\frac{\Pr{[Y = y]}}{\Pr{[Y = y + k^\prime]}} &= \Big(\frac{1-p}{p}\Big)^{k^\prime} \prod_{i=1}^{k^\prime} \frac{y + i}{n-y -i+1} \\
&\leq \Big(\frac{1-p}{p}\Big)^{k^\prime} \Big(\frac{y + k^\prime}{n-y -k^\prime}\Big)^{k^\prime} \\ 
&\leq \Big(\frac{1-p}{p}\Big)^{k^\prime} \Big(\frac{(1 + \alpha)np}{n- (1 + \alpha)np}\Big)^{k^\prime} \label{eq: plug_in}\\
&= (1 + \alpha)^{k^\prime} \Big(\frac{1-p}{1- p - p\alpha}\Big)^{k^\prime} \label{eq: lhs}
\end{align}

$\eqref{eq: plug_in}$ comes from our assumption $y \in \mathcal{\epsilon}$ and so when $y=(1 + \alpha)np - k^\prime$ the ratio above is maximal.

Similarly for $-k \leq k^\prime \leq 0$ we have 
\begin{align}
\frac{\Pr{[Y = y]}}{\Pr{[Y = y + |k^\prime |]}} &= \Big(\frac{p}{1-p}\Big)^{|k^\prime|} \prod_{i=1}^{|k^\prime|} \frac{n-y +i}{y - i + 1} \\
&\leq \Big(\frac{p}{1-p}\Big)^{|k^\prime|} \Big( \frac{n-y +|k^\prime|}{y - |k^\prime|}\Big)^{|k^\prime|} \\
&\leq \Big( \frac{1 + p\alpha - p}{(1-\alpha)(1 - p)}\Big)^{|k^\prime|} \label{eq: plug_in2} \\
&\leq \Big(\frac{1 + \alpha}{1 - \alpha}\Big)^{|k^\prime|} \label{eq: p_restriction} \\
&= e^{|k|\epsilon} 
\end{align}

$\eqref{eq: plug_in2}$ comes from plugging in the smallest value for $y$ and \eqref{eq: p_restriction} comes from the fact that $p \leq 1/2$. 
It is easy to see that when $p \leq \frac{1}{2}$, that we have 

\begin{equation}
 \Big(\frac{1-p}{1- p - p\alpha}\Big)^{k^\prime} \leq 	\Big(\frac{1}{1 - \alpha}\Big)^{|k^\prime|} 
\end{equation}

Therefore we can upper bound equation \eqref{eq: lhs} by setting 
\begin{equation}
(1 + \alpha)^{k^\prime} \Big(\frac{1-p}{1- p - p\alpha}\Big)^{k^\prime} \leq  e^{|k|\epsilon} 
\end{equation}

Finally, we can prove smoothness using Bayes' rule. 
Let $g(y) = \frac{\Pr_{Y \sim \Bin(n,p)}[Y = y]}{\Pr_{Y \sim \Bin(n, p)}[Y = y+k^{\prime}]}$ and $\mathcal{D}=\Bin(n, p)$.
Then:
\begin{equation}
\begin{split}
\Pr_{y \sim \mathcal{D}} \Big[ g(y) \geq e^{|k^{\prime}|\epsilon}\Big] 
&\leq \Pr_{y \sim \mathcal{D}} \Big[ g(y) \geq e^{|k^{\prime}|\epsilon} \Big| y \notin \mathcal{\epsilon} \Big] + \delta \\
&\leq e^{-\frac{-\alpha^2np}{8}} + e^{-\frac{-\alpha^2np}{8+2\alpha}}  \label{eq: y_condition} \\
&\leq \delta
\end{split}
\end{equation}

$\eqref{eq: y_condition}$ comes from the fact that when $y \in \mathcal{\epsilon}$ we have $\frac{\Pr{[Y = y]}}{\Pr{[Y = y + k^\prime]}} \leq e^{\epsilon}|k^\prime|$ by how we defined $\epsilon$ in equation \eqref{eq: p_restriction}. So we have 

$\Pr_{y \sim \Bin(n, p)} \Bigg[ \frac{\Pr_{Y \sim \Bin(n,p)}[Y = y]}{\Pr_{Y \sim \Bin(n, p)}[Y = y+k^{\prime}]} \geq e^{|k^{\prime}|\epsilon} \Big| y \notin \mathcal{\epsilon} \Bigg] = 0$
\end{proof}

Putting all this together, we can now prove Lemma~\ref{theorem:dp_guarantee}: 

\begin{proof}[Proof of Lemma~\ref{theorem:dp_guarantee}]
From the definition of $k$-incremental functions in Definition \ref{def: k-incremental} of Appendix~\ref{app:dp_proof_bin_mech}, it is easy to see that the counting query is 1-incremental and has global sensitivity $\Delta=1$. So, by Lemma \ref{lemma:smooth_distributions} in Appendix~\ref{app:dp_proof_bin_mech} if we add noise $Z \sim \mathcal{D}$ where $\mathcal{D}$ is $(\epsilon, \delta, 1)$-smooth we get our result.  
Thus if we could show that $\Bin(\noisen, \frac{1}{2})$ is $(\epsilon, \delta, 1)$-smooth we would be done. 

From Lemma \ref{lemma: bin_is_k_smooth} in Appendix~\ref{app:dp_proof_bin_mech} we have that for $\alpha = \frac{e^\epsilon -1 }{e^\epsilon + 1}$ and $k \leq \frac{\noisen\alpha p}{2}$, where $p \leq \frac12$, the binomial distribution $\Bin(\noisen, p)$ is $(\epsilon, \delta, k)$-smooth, where $\delta \geq e^{-\frac{-\alpha^2 \noisen p}{8}} + e^{-\frac{-\alpha^2 \noisen p}{8+2\alpha}}$. 
Observe that by the definition of $\alpha$, for all $\epsilon \in [0, 1]$, $\alpha \geq \frac{\epsilon}{\sqrt{5}}$. 
Setting $p=\frac12$, we could write 
\begin{align}
\noisen p\alpha  \geq \noisen\frac{\epsilon}{2\sqrt{5}} \label{bin_1_smooth}
\end{align}

All we need is an adequate value for $\epsilon \in [0, 1]$ by which we get $\noisen p\alpha \geq 2$, which would prove that $\Bin(\noisen, \frac12)$ is 1-smooth. 
Notice that
\begin{align}
e^{-\frac{-\alpha^2\noisen p}{8}} + e^{-\frac{-\alpha^2 \noisen p}{8+2\alpha}} &\leq 2e^{-\frac{-\alpha^2 \noisen p}{10}} 
\leq 2e^{-\frac{-\epsilon^2 \noisen p}{50}}  = \delta
\end{align}
Re-arranging the terms and setting $p=\frac12$ we get $\epsilon = 10\sqrt{\frac{1}{\noisen}\ln\frac{2}{\delta}}$. 
Plugging it back into equation \eqref{bin_1_smooth} we get 
\begin{align*}
\noisen p\frac{\epsilon}{\sqrt{5}} = \frac{\epsilon}{2\sqrt{5}} 
= \noisen\sqrt{\frac{1}{5 \noisen}\ln\frac{2}{\delta}} 
& \geq \sqrt{\frac{\noisen}{5}\ln 2} \\
& \geq 2 \text{ for $\noisen > 30$ }
\end{align*} 

Therefore we have shown that $\Bin(\noisen, \frac12)$ is $(\epsilon, \delta, 1)$-smooth
\end{proof}

%% file: Appendix/c_fiat_shamir.tex
\begin{figure*}[!h]
    \centering
\begin{pchstack}[boxed]  
    \pseudocode[linenumbering , skipfirstln]{%
    \textbf{Verifier} \< \< \textbf{Prover} \\[][\hline]
     \< \text{ Common Input } g, h, \G_q, q, c \<    \\
    \<\< v_1, e_1 \xleftarrow{R} \Z_q^2 \\
    \<\< \text{Set } d_1 \text{ such that }d_1\Big(\frac{c}{g}\Big)^{e_1} = h^{v_1} \\
    \<\< b \xleftarrow{R} \Z_q \text{ and set } d_0 = h^b\\
    (d_0, d_1) \< \sendmessageleft*{(d_0, d_1)} \<(d_0, d_1) \\
    e \xleftarrow{R} \Z_q \< \sendmessageright*{e} \< e_0 = e - e_1 \mod q \\
    \<\< v_0 = b + e_0r \\
    \text{Check } e_1 + e_0 = e \< \sendmessageleft*{(v_0, e_0, v_1, e_1)} \< \\
    \text{Check }d_0c^{e_0} = h^{v_0} \text{ and } d_1c^{e_1} = g^{e_1}h^{v_1}\<\<
}
\end{pchstack}    
    \caption{Proof for convincing $\Vfr$ that $c = gh^{r}$ is in $L_{\texttt{Bit}}$ without revealing that $x=1$.}
    \label{fig:schnorr_or_a}
\end{figure*}

\begin{figure*}[!h]
    \centering
\begin{pchstack}[boxed]  
    \pseudocode[linenumbering , skipfirstln]{%
    \textbf{Verifier} \< \< \textbf{Prover} \\[][\hline]
     \< \text{ Common Input } g, h, \G_q, q, c \<    \\
    \<\< v_0, e_0 \xleftarrow{R} \Z_q^2 \\
    \<\< \text{Set } d_0 \text{ such that }d_0c^{e_0} = h^{v_0} \\
    \<\< b \xleftarrow{R} \Z_q \text{ and set } d_1 = h^b\\
    (d_0, d_1) \< \sendmessageleft*{(d_0, d_1)} \<(d_0, d_1) \\
    e \xleftarrow{R} \Z_q \< \sendmessageright*{e} \< e_1 = e - e_0 \mod q \\
    \<\< v_1 = b + e_1r \\
    \text{Check } e_1 + e_0 = e \< \sendmessageleft*{(v_0, e_0, v_1, e_1)} \< \\
    \text{Check }d_0c^{e_0} = h^{v_0} \text{ and } d_1c^{e_1} = g^{e_1}h^{v_1}\<\<
}
\end{pchstack}    
    \caption{Proof for convincing $\Vfr$ that $c_x = h^{r_x}$ is in $L_{\texttt{Bit}}$ without revealing that $x=0$.}
    \label{fig:schnorr_or_b}
\end{figure*}

\section{OR Protocol}
\label{app:sigma_open}

Define as public parameters a cyclic prime order group $\G_q$ and generators $g$ and $h$ for $\G_q$. Let $\mathcal{M}_\pp = \mathcal{R}_\pp = \Z_q$. Pedersen Commitments defined below satisfy all the properties described in Section~\ref{sec: commitments}.
\begin{equation}
    \Com(x, r_x) = g^xh^{r_x}
\end{equation}

For the sake of completeness, we describe the interactive disjunctive OR proof using $\Sigma$-protocols from \cite{cramer1994proofs}. In all implementations in this paper, we use the Fiat-Shamir transform, which makes protocols non-interactive and can be shown to be secure in the random oracle model as described in \cite{thaler2020proofs}. 
Note we choose the non-interactive version for efficiency reasons but the $\Sigma$ protocols are zero-knowledge even without a random oracle. Maurer \cite{maurer2009unifying} shows that if the verifier's challenge space is polynomial sized then the protocol can be shown to be zero-knowledge. 
This does affect the soundness of the protocol but it can be made negligible by a constant number of repetitions.
Damgard \etal show that by using Trapdoor commitments \cite{damgaard2000efficient}, one can preserve soundness and get zero-knowledge but the protocol now has 4 messaging rounds instead of 3. Next we describe the $\Sigma$-protocol that can used to verify the OR condition. \par 
Let $x \in \bit$ and $c_x = \Com(x, r_x)$ for $r_x \xleftarrow{R} \Z_q$ be the commitment to $x$. Given $c_x$, $\Pi_{\texttt{OR}}$ is an interactive zero-knowledge proof between a prover $\Pv$ and a verifier $\Vfr$ to show that $c_x \in L_{\texttt{Bit}}$. The security properties can be found in \cite{thaler2020proofs, damgaard2000efficient, cramer1994proofs}. 
\begin{equation}
L_{\texttt{Bit}} = \{c_x: x \in \bit \land c_x = \Com(x, r_x) \}    
\end{equation}
where for some $r_x \in \Z_q$. Based on the value for $x$, the prover uses the protocol described in Figure \ref{fig:schnorr_or_a} or Figure \ref{fig:schnorr_or_b}. The verifier cannot distinguish between the protocol's two runs as the messages are indistinguishable. In case the the inputs $\vec{x}$ are bit strings of size $M$ (like in PRIO and Poplar) and only one coordinate can be non-negative, the prover sends to the verifier $r = \sum_{j=1}^M r_{x_j}$ along with the $\Sigma$-proofs, where $r_{x_j}$ is the randomness used to create commitments for the coordinate $x_j \in \bit$. As $\vec{x} \in L$, implies $\vec{x} \in \bit^M$ and $|| \vec{x}||_1 = 1$, the OR proofs verify the first criterion and the second criterion is easily verified by checking $c_{|| \vec{x}||} = \prod_{j=1}^M c_{x_m}$ is a a commitment to one i.e., check if $g^1h^r = c$. 

%% file: Appendix/d_mpc_defs.tex
\section{Deferred Security Proofs}
\label{app:single_curator_sec_proof}

\parag{Single Curator Simulator Proof}

\begin{theorem}
Let  $\Vfr^*$ denote the corrupted verifier. There exists a PPT Simulator $\Sim_{(\Vfr^*)}$ such that for all $y = \mathcal{M}_\bin(X, Q)$
 \begin{align*}
    \texttt{View}\left[\Pi(\Pv, \Vfr)\right] &\stackrel{c}{\equiv} \Sim_{(\Vfr^*}(y, \bvec{r}_v, z, \pp)        
    \end{align*}

    where $z \in \bit^{\texttt{poly}(\kappa)}$ and $\bvec{r}_v \in \bit^{\texttt{poly}(\kappa)}$ represents auxiliary input and randomness available to all the corrupted parties.  	
\end{theorem}
\begin{proof}
Denote the corrupted verifier as $\Vfr^*$.  $\Sim$ receives on its input tape the inputs for $\Vfr^*$. The ideal oracle functionality $\mathcal{M}_\bin$ is defined as before. Let $\Sim$ denote shorthand for $\Sim_{\Vfr^*}$.  We construct the simulator as follows:
\begin{enumerate}    
    \item{$\Sim$ receives the public messages $\{c_i\}_{i \in [n]}$.}

    \item{$\Sim$ invokes $\mathcal{M}_{\texttt{bin}}$ with the empty string $\lambda$ and receives $y$ as defined in equation \eqref{eq:M_bin}.}
    
    \item{$\Sim$ samples $z \xleftarrow[]{R} \mathcal{R}_\pp$ and sets $c = \Com(y, z)$.}

    \item{$\Sim$ samples $c'_{2}, \dots, c'_{{\noisen}}$ such that $c'_{j} = \Com(1, s_{j})$ where $s_{j} \xleftarrow[]{R} \mathcal{R}_\pp$. 
     It sets $c'_{1} = g^1a$ where $a = c \times  \Big( \prod_{j=2}^{\noisen} \hat{c}'_{j}\Big)^{-1} \times \Big( \prod_{i=1}^{n} c_{i}\Big)^{-1} \times g^{-1}$.}
    
     \item{$\Sim$ sends over $\{c_{j}\}_{j \in [\noisen]}$ to $\AdvA$ pretending to be the honest prover (Line 4 of Figure \ref{fig:dp_interactive_proof}).}
	    
    \item{$\Sim$ pretends to be the prover and jointly invokes $\oracle_{\texttt{Morra}}$ with $\AdvA$ to sample $\noisen$ unbiased public bits $(b_{1}, \dots, b_{\noisen})$.}
    
    \item{$\Sim$ sends $y$ and $z$ to $\AdvA$ and outputs whatever $\AdvA$ outputs.}\end{enumerate}\end{proof}

%% file: sample-sigconf.bbl
\newcommand{\etalchar}[1]{$^{#1}$}
\begin{thebibliography}{RCWH{\etalchar{+}}20}

\bibitem[AGJ{\etalchar{+}}22]{addanki2022prio+}
Surya Addanki, Kevin Garbe, Eli Jaffe, Rafail Ostrovsky, and Antigoni
  Polychroniadou.
\newblock Prio+: Privacy preserving aggregate statistics via boolean shares.
\newblock In {\em International Conference on Security and Cryptography for
  Networks}, pages 516--539. Springer, 2022.

\bibitem[AJL04]{ambainis2004cryptographic}
Andris Ambainis, Markus Jakobsson, and Helger Lipmaa.
\newblock Cryptographic randomized response techniques.
\newblock In {\em International Workshop on Public Key Cryptography}, pages
  425--438. Springer, 2004.

\bibitem[BBB{\etalchar{+}}18]{bunz2018bulletproofs}
Benedikt B{\"u}nz, Jonathan Bootle, Dan Boneh, Andrew Poelstra, Pieter Wuille,
  and Greg Maxwell.
\newblock Bulletproofs: Short proofs for confidential transactions and more.
\newblock In {\em 2018 IEEE symposium on security and privacy (SP)}, pages
  315--334. IEEE, 2018.

\bibitem[BBCG{\etalchar{+}}22]{boneh_lightweight_2022}
Dan Boneh, Elette Boyle, Henry Corrigan-Gibbs, Niv Gilboa, and Yuval Ishai.
\newblock Lightweight {Techniques} for {Private} {Heavy} {Hitters}.
\newblock {\em arXiv:2012.14884 [cs]}, January 2022.
\newblock arXiv: 2012.14884.

\bibitem[BBG{\etalchar{+}}20]{bell2020secure}
James~Henry Bell, Kallista~A Bonawitz, Adri{\`a} Gasc{\'o}n, Tancr{\`e}de
  Lepoint, and Mariana Raykova.
\newblock Secure single-server aggregation with (poly) logarithmic overhead.
\newblock In {\em Proceedings of the 2020 ACM SIGSAC Conference on Computer and
  Communications Security}, pages 1253--1269, 2020.

\bibitem[BBGN19]{balle2019privacy}
Borja Balle, James Bell, Adri{\`a} Gasc{\'o}n, and Kobbi Nissim.
\newblock The privacy blanket of the shuffle model.
\newblock In {\em Annual International Cryptology Conference}, pages 638--667.
  Springer, 2019.

\bibitem[BC20]{balcer_separating_2020}
Victor Balcer and Albert Cheu.
\newblock Separating {Local} \& {Shuffled} {Differential} {Privacy} via
  {Histograms}.
\newblock {\em arXiv:1911.06879 [cs]}, April 2020.
\newblock arXiv: 1911.06879.

\bibitem[BCV16]{bun2016separating}
Mark Bun, Yi-Hsiu Chen, and Salil Vadhan.
\newblock Separating computational and statistical differential privacy in the
  client-server model.
\newblock In {\em Theory of Cryptography Conference}, pages 607--634. Springer,
  2016.

\bibitem[BDO14]{baum2014publicly}
Carsten Baum, Ivan Damg{\aa}rd, and Claudio Orlandi.
\newblock Publicly auditable secure multi-party computation.
\newblock In {\em International Conference on Security and Cryptography for
  Networks}, pages 175--196. Springer, 2014.

\bibitem[BGG{\etalchar{+}}22]{bell2022distributed}
James Bell, Adria Gascon, Badih Ghazi, Ravi Kumar, Pasin Manurangsi, Mariana
  Raykova, and Phillipp Schoppmann.
\newblock Distributed, private, sparse histograms in the two-server model.
\newblock {\em Cryptology ePrint Archive}, 2022.

\bibitem[BGI16]{boyle_function_2016}
Elette Boyle, Niv Gilboa, and Yuval Ishai.
\newblock Function {Secret} {Sharing}: {Improvements} and {Extensions}.
\newblock In {\em Proceedings of the 2016 {ACM} {SIGSAC} {Conference} on
  {Computer} and {Communications} {Security}}, pages 1292--1303, Vienna
  Austria, October 2016. ACM.

\bibitem[BGI19]{boyle2019secure}
Elette Boyle, Niv Gilboa, and Yuval Ishai.
\newblock Secure computation with preprocessing via function secret sharing.
\newblock In {\em Theory of Cryptography Conference}, pages 341--371. Springer,
  2019.

\bibitem[BK07]{bell2007lessons}
Robert~M Bell and Yehuda Koren.
\newblock Lessons from the netflix prize challenge.
\newblock {\em Acm Sigkdd Explorations Newsletter}, 9(2):75--79, 2007.

\bibitem[BK21]{bohler2021secure}
Jonas B{\"o}hler and Florian Kerschbaum.
\newblock Secure multi-party computation of differentially private heavy
  hitters.
\newblock In {\em Proceedings of the 2021 ACM SIGSAC Conference on Computer and
  Communications Security}, pages 2361--2377, 2021.

\bibitem[Blu83]{blum1983coin}
Manuel Blum.
\newblock Coin flipping by telephone a protocol for solving impossible
  problems.
\newblock {\em ACM SIGACT News}, 15(1):23--27, 1983.

\bibitem[BMRS21]{baum2021mathsf}
Carsten Baum, Alex~J Malozemoff, Marc~B Rosen, and Peter Scholl.
\newblock Mac'n'cheese : Zero-knowledge proofs for boolean and arithmetic
  circuits with nested disjunctions.
\newblock In {\em Annual International Cryptology Conference}, pages 92--122.
  Springer, 2021.

\bibitem[BOGW19]{ben2019completeness}
Michael Ben-Or, Shafi Goldwasser, and Avi Wigderson.
\newblock Completeness theorems for non-cryptographic fault-tolerant
  distributed computation.
\newblock In {\em Providing Sound Foundations for Cryptography: On the Work of
  Shafi Goldwasser and Silvio Micali}, pages 351--371. 2019.

\bibitem[BS22]{boyd2022differential}
Danah Boyd and Jayshree Sarathy.
\newblock Differential perspectives: Epistemic disconnects surrounding the us
  census bureau’s use of differential privacy.
\newblock {\em Harvard Data Science Review (Forthcoming)}, 2022.

\bibitem[CDS94]{cramer1994proofs}
Ronald Cramer, Ivan Damg{\aa}rd, and Berry Schoenmakers.
\newblock Proofs of partial knowledge and simplified design of witness hiding
  protocols.
\newblock In {\em Annual International Cryptology Conference}, pages 174--187.
  Springer, 1994.

\bibitem[CGB17]{corrigan-gibbs_prio_2017}
Henry Corrigan-Gibbs and Dan Boneh.
\newblock Prio: {Private}, {Robust}, and {Scalable} {Computation} of
  {Aggregate} {Statistics}.
\newblock {\em arXiv:1703.06255 [cs]}, March 2017.
\newblock arXiv: 1703.06255.

\bibitem[Che21]{cheu2021differential}
Albert Cheu.
\newblock Differential privacy in the shuffle model: A survey of separations.
\newblock {\em arXiv preprint arXiv:2107.11839}, 2021.

\bibitem[Cle86]{cleve1986limits}
Richard Cleve.
\newblock Limits on the security of coin flips when half the processors are
  faulty.
\newblock In {\em Proceedings of the eighteenth annual ACM symposium on Theory
  of computing}, pages 364--369, 1986.

\bibitem[CSU19]{champion2019securely}
Jeffrey Champion, Abhi Shelat, and Jonathan Ullman.
\newblock Securely sampling biased coins with applications to differential
  privacy.
\newblock In {\em Proceedings of the 2019 ACM SIGSAC Conference on Computer and
  Communications Security}, pages 603--614, 2019.

\bibitem[CSU21]{cheu2021manipulation}
Albert Cheu, Adam Smith, and Jonathan Ullman.
\newblock Manipulation attacks in local differential privacy.
\newblock In {\em 2021 IEEE Symposium on Security and Privacy (SP)}, pages
  883--900. IEEE, 2021.

\bibitem[Dam00]{damgaard2000efficient}
Ivan Damg{\aa}rd.
\newblock Efficient concurrent zero-knowledge in the auxiliary string model.
\newblock In {\em International Conference on the Theory and Applications of
  Cryptographic Techniques}, pages 418--430. Springer, 2000.

\bibitem[DIO20]{dittmer2020line}
Samuel Dittmer, Yuval Ishai, and Rafail Ostrovsky.
\newblock Line-point zero knowledge and its applications.
\newblock {\em Cryptology ePrint Archive}, 2020.

\bibitem[DKM{\etalchar{+}}06]{dwork2006our}
Cynthia Dwork, Krishnaram Kenthapadi, Frank McSherry, Ilya Mironov, and Moni
  Naor.
\newblock Our data, ourselves: Privacy via distributed noise generation.
\newblock In {\em Annual international conference on the theory and
  applications of cryptographic techniques}, pages 486--503. Springer, 2006.

\bibitem[DKS{\etalchar{+}}21]{ding2021permute}
Zeyu Ding, Daniel Kifer, Thomas Steinke, Yuxin Wang, Yingtai Xiao, Danfeng
  Zhang, et~al.
\newblock The permute-and-flip mechanism is identical to report-noisy-max with
  exponential noise.
\newblock {\em arXiv preprint arXiv:2105.07260}, 2021.

\bibitem[DMNS06]{dwork2006calibrating}
Cynthia Dwork, Frank McSherry, Kobbi Nissim, and Adam Smith.
\newblock Calibrating noise to sensitivity in private data analysis.
\newblock In {\em Theory of cryptography conference}, pages 265--284. Springer,
  2006.

\bibitem[DN00]{dwork2000zaps}
Cynthia Dwork and Moni Naor.
\newblock Zaps and their applications.
\newblock In {\em Proceedings 41st Annual Symposium on Foundations of Computer
  Science}, pages 283--293. IEEE, 2000.

\bibitem[DSQ{\etalchar{+}}21]{davidson2021star}
Alex Davidson, Peter Snyder, EB~Quirk, Joseph Genereux, and Benjamin Livshits.
\newblock Star: Distributed secret sharing for private threshold aggregation
  reporting.
\newblock {\em arXiv preprint arXiv:2109.10074}, 2021.

\bibitem[EFM{\etalchar{+}}20]{erlingsson_amplification_2020}
Úlfar Erlingsson, Vitaly Feldman, Ilya Mironov, Ananth Raghunathan, Kunal
  Talwar, and Abhradeep Thakurta.
\newblock Amplification by {Shuffling}: {From} {Local} to {Central}
  {Differential} {Privacy} via {Anonymity}.
\newblock {\em arXiv:1811.12469 [cs, stat]}, July 2020.
\newblock arXiv: 1811.12469.

\bibitem[GAM19]{garfinkel2019understanding}
Simson Garfinkel, John~M Abowd, and Christian Martindale.
\newblock Understanding database reconstruction attacks on public data.
\newblock {\em Communications of the ACM}, 62(3):46--53, 2019.

\bibitem[GGK{\etalchar{+}}19]{ghazi2019power}
Badih Ghazi, Noah Golowich, Ravi Kumar, Rasmus Pagh, and Ameya Velingker.
\newblock On the power of multiple anonymous messages.
\newblock {\em arXiv preprint arXiv:1908.11358}, 2019.

\bibitem[GGK{\etalchar{+}}20]{ghazi_power_2020}
Badih Ghazi, Noah Golowich, Ravi Kumar, Rasmus Pagh, and Ameya Velingker.
\newblock On the {Power} of {Multiple} {Anonymous} {Messages}.
\newblock {\em arXiv:1908.11358 [cs, stat]}, May 2020.
\newblock arXiv: 1908.11358.

\bibitem[GKY11]{groce2011limits}
Adam Groce, Jonathan Katz, and Arkady Yerukhimovich.
\newblock Limits of computational differential privacy in the client/server
  setting.
\newblock In {\em Theory of Cryptography Conference}, pages 417--431. Springer,
  2011.

\bibitem[Gol07]{goldreich_foundations_2007}
Oded Goldreich.
\newblock {\em Foundations of cryptography. {Vol}. 1: {Basic} tools}, volume~1.
\newblock Cambridge Univ. Press, Cambridge, digitally print. 1. paperback
  version edition, 2007.

\bibitem[HO14]{haitner2014coin}
Iftach Haitner and Eran Omri.
\newblock Coin flipping with constant bias implies one-way functions.
\newblock {\em SIAM Journal on Computing}, 43(2):389--409, 2014.

\bibitem[JL09]{jarecki2009efficient}
Stanis{\l}aw Jarecki and Xiaomin Liu.
\newblock Efficient oblivious pseudorandom function with applications to
  adaptive ot and secure computation of set intersection.
\newblock In {\em Theory of Cryptography Conference}, pages 577--594. Springer,
  2009.

\bibitem[Jus21]{courtCase}
Brennan Center~For Justice.
\newblock Alabama v. u.s. dep’t of commerce, 2021.

\bibitem[KCY21]{kato2021preventing}
Fumiyuki Kato, Yang Cao, and Masatoshi Yoshikawa.
\newblock Preventing manipulation attack in local differential privacy using
  verifiable randomization mechanism.
\newblock In {\em IFIP Annual Conference on Data and Applications Security and
  Privacy}, pages 43--60. Springer, 2021.

\bibitem[KLN{\etalchar{+}}11]{kasiviswanathan2011can}
Shiva~Prasad Kasiviswanathan, Homin~K Lee, Kobbi Nissim, Sofya Raskhodnikova,
  and Adam Smith.
\newblock What can we learn privately?
\newblock {\em SIAM Journal on Computing}, 40(3):793--826, 2011.

\bibitem[KRS13]{kasiviswanathan2013power}
Shiva~Prasad Kasiviswanathan, Mark Rudelson, and Adam Smith.
\newblock The power of linear reconstruction attacks.
\newblock In {\em Proceedings of the twenty-fourth annual ACM-SIAM symposium on
  Discrete algorithms}, pages 1415--1433. SIAM, 2013.

\bibitem[Lin17]{lindell2017simulate}
Yehuda Lindell.
\newblock How to simulate it--a tutorial on the simulation proof technique.
\newblock {\em Tutorials on the Foundations of Cryptography}, pages 277--346,
  2017.

\bibitem[Mau09]{maurer2009unifying}
Ueli Maurer.
\newblock Unifying zero-knowledge proofs of knowledge.
\newblock In {\em International Conference on Cryptology in Africa}, pages
  272--286. Springer, 2009.

\bibitem[MMP{\etalchar{+}}10]{mcgregor2010limits}
Andrew McGregor, Ilya Mironov, Toniann Pitassi, Omer Reingold, Kunal Talwar,
  and Salil Vadhan.
\newblock The limits of two-party differential privacy.
\newblock In {\em 2010 IEEE 51st Annual Symposium on Foundations of Computer
  Science}, pages 81--90. IEEE, 2010.

\bibitem[MPRV09]{mironov2009computational}
Ilya Mironov, Omkant Pandey, Omer Reingold, and Salil Vadhan.
\newblock Computational differential privacy.
\newblock In {\em Annual International Cryptology Conference}, pages 126--142.
  Springer, 2009.

\bibitem[MT07]{mcsherry2007mechanism}
Frank McSherry and Kunal Talwar.
\newblock Mechanism design via differential privacy.
\newblock In {\em 48th Annual IEEE Symposium on Foundations of Computer Science
  (FOCS'07)}, pages 94--103. IEEE, 2007.

\bibitem[Ped91]{pedersen1991non}
Torben~Pryds Pedersen.
\newblock Non-interactive and information-theoretic secure verifiable secret
  sharing.
\newblock In {\em Annual international cryptology conference}, pages 129--140.
  Springer, 1991.

\bibitem[RCWH{\etalchar{+}}20]{roy2020crypt}
Amrita Roy~Chowdhury, Chenghong Wang, Xi~He, Ashwin Machanavajjhala, and Somesh
  Jha.
\newblock Cryptϵ: Crypto-assisted differential privacy on untrusted servers.
\newblock In {\em Proceedings of the 2020 ACM SIGMOD International Conference
  on Management of Data}, pages 603--619, 2020.

\bibitem[RHML21]{raturi2021impact}
Varun Raturi, Jinhyun Hong, David~Philip McArthur, and Mark Livingston.
\newblock The impact of privacy protection measures on the utility of
  crowdsourced cycling data.
\newblock {\em Journal of Transport Geography}, 92:103020, 2021.

\bibitem[Sha79]{shamir1979share}
Adi Shamir.
\newblock How to share a secret.
\newblock {\em Communications of the ACM}, 22(11):612--613, 1979.

\bibitem[Tha20]{thaler2020proofs}
Justin Thaler.
\newblock Proofs, arguments, and zero-knowledge, 2020.

\bibitem[Vad17]{vadhan2017complexity}
Salil Vadhan.
\newblock The complexity of differential privacy.
\newblock In {\em Tutorials on the Foundations of Cryptography}, pages
  347--450. Springer, 2017.

\bibitem[War65]{warner1965randomized}
Stanley~L Warner.
\newblock Randomized response: A survey technique for eliminating evasive
  answer bias.
\newblock {\em Journal of the American Statistical Association},
  60(309):63--69, 1965.

\bibitem[WYKW21]{weng2021wolverine}
Chenkai Weng, Kang Yang, Jonathan Katz, and Xiao Wang.
\newblock Wolverine: fast, scalable, and communication-efficient zero-knowledge
  proofs for boolean and arithmetic circuits.
\newblock In {\em 2021 IEEE Symposium on Security and Privacy (SP)}, pages
  1074--1091. IEEE, 2021.

\end{thebibliography}
